\documentclass[11pt]{article}

\usepackage[margin=1in]{geometry}
\usepackage{enumerate}
\usepackage{amsmath}
\usepackage{amssymb}
\usepackage{amsthm}
\usepackage{graphicx}
\usepackage{dsfont}
\usepackage{comment}
\usepackage[colorlinks=true, allcolors=blue]{hyperref}
\newtheorem{theorem}{Theorem}[section]
\newtheorem{lemma}{Lemma}[section]

\newtheorem{definition}{Definition}[section]
\newtheorem{remark}{Remark}[section]

\usepackage{booktabs}
\usepackage{xcolor}
\usepackage{bm}
\usepackage{mathrsfs}

\def \E{\mathbb{E}}

\def \N{\mathbb{N}}
\def \P{\mathbb{P}}
\def \Q{\mathbb{Q}}
\def \R{\mathbb{R}}

\def \cP{\mathcal{P}}

\def \fX{\mathfrak{X}}
\def \fm{\mathfrak{m}}
\def \fs{\mathfrak{s}}

\def \cJ{\mathcal{J}}

\def \cZ{\mathcal{Z}}
\def \cI{\mathcal{I}}
\def \D{\mathcal{D}}
\def \cF{\mathcal{F}}
\def \cV{\mathcal{V}}

\def \A{\mathcal{A}}
\def \sfA{\mathsf{A}}
\def \sfB{\mathsf{B}}
\def \sfC{\mathsf{C}}
\def \sfD{\mathsf{D}}
\def \cB{\mathcal{B}}
\def \eps{\varepsilon}

\usepackage{enumitem}

\numberwithin{equation}{section}

\title{Mean-Variance Stackelberg Games with Asymmetric Information}
\author{
    Yu-Jui Huang\thanks{Department of Applied Mathematics, University of Colorado, Boulder, CO 80309-0526, USA, email: \texttt{yujui.huang@colorado.edu}. Partially supported by National Science Foundation (DMS-2109002).}
    \and
   Shihao Zhu\thanks{Institute of Insurance Science, Ulm University, Helmholtzstr.\ 20, 89069 Ulm, Germany, email: \texttt{shihao.zhu@uni-ulm.de}. Partially supported by the Deutsche Forschungsgemeinschaft (DFG), Project-ID 509303834, FOR 5583 ``Asset Allocation and Asset Pricing under Regulatory Uncertainty".}
}

\date{\today}

\begin{document}

\maketitle


\begin{abstract}
Two investors perform mean-variance portfolio selection under asymmetric information: investor 1 knows the true stock dynamics; investor 2 has to infer the true dynamics from observed stock evolution. Their portfolio choices are connected by 
relative performance concerns, i.e., each investor is concerned about not only the individual terminal wealth, but also how it compares to the average terminal wealth of both investors. We model this as Stackelberg competition: investor 2 (the ``follower'') observes actual trading of investor 1 and decides his trading strategy accordingly; anticipating investor 2's reaction, investor 1 (the ``leader'') selects a trading strategy that best suits her objective. As investor 2 only observes investor 1's actual trading (without knowing its dynamics), he performs {\it pathwise} optimization, i.e., solving his mean-variance problem under every fixed path of trading actions investor 1 might take, which yields an intrapersonal equilibrium strategy depending linearly on investor 1's real-time trading. Anticipating investor 2's use of this strategy, investor 1 aims to randomize her trading actions (to better conceal private information from investor 2) 
under an {\it entropy-regularized} mean-variance objective, where actions are randomized realistically on a time grid (which can be however dense, but remains discrete) and an entropy term quantifies the induced randomness. By exploratory stochastic control and Malliavin calculus, we obtain investor 1's intrapersonal equilibrium strategy, a Gaussian policy integrating her private information and internalization of investor 2's strategy. The two investors' respective strategies then form a Stackelberg equilibrium between them. 
 \end{abstract}

\vspace{2mm}

\vspace{2mm}
\vspace{2mm}


\section{Introduction}
Investors' trading strategies are often intertwined. In making their trading decisions, since investors may possess different levels of information about the stock dynamics, a less-informed one may hope to learn from the observed trading of a well-informed one, who may in turn trade cautiously to conceal her private information. When evaluating their trading decisions, investors not only consider their own investment performance, but commonly compare it with the performance of others. This paper aims to understand how these two factors, {\it asymmetric information} and {\it relative evaluation}, jointly affect investors' trading decisions.

We consider two investors who trade a stock \(S\) over a finite time horizon \(T>0\). For concreteness, we assume that the expected return of \(S\) is a constant \(\mu\in\mathbb R\), known to investor 1. Investor 2 does not know \(\mu\) precisely, except that it has two possible values \(\mu_1,\mu_2\in\mathbb R\) (with \(\mu_1>\mu_2\)). Relying on stochastic filtering, investor 2 infers the true dynamics of $S$ through the posterior probability \(P(t) := \P(\mu=\mu_1\mid \mathcal F^S_t)\), where $\cF^S_\cdot$ is the filtration obtained from observed stock evolution. The two investors' portfolio selection problems, formulated under the true dynamics and the filtered dynamics, respectively, are linked via their \emph{relative performance concerns}. Investor \(i\) (for \(i=1,2\)) is concerned about not only the individual terminal wealth \(X_i(T)\), but also how it compares with the average terminal wealth $\overline X(T):={(X_1(T)+X_2(T))}/2$, thereby considering the {\it mixed} terminal wealth 
$(1-\lambda_i)X_i(T) + \lambda_i\bigl(X_i(T)-\overline X(T)\bigr)= X_i(T)-\lambda_i\overline X(T)$ for some \(\lambda_i\in[0,1)\). This modeling of relative performance concerns aligns with \cite{espinosa2015optimal, lacker2019mean, huang2026partial}, albeit investors therein share the same information about the stock dynamics: all are fully-informed in \cite{espinosa2015optimal, lacker2019mean} under the expected utility paradigm; all are partially-informed in \cite{huang2026partial} under a mean-variance criterion. 

In this paper, we integrate asymmetric information into the mean-variance portfolio selection game in \cite{huang2026partial} using Stackelberg's leader-follower structure: investor 2 (the ``follower'') observes actual trading of investor 1 and decides
his trading strategy accordingly; anticipating investor 2's reaction, investor 1 (the ``leader'') selects a trading strategy that best suits her objective. As is well-known in game theory, an informed player may randomize her actions to conceal her private information from others who observe her actions; see e.g., \cite{cardaliaguet2007differential, cardaliaguet2009stochastic, grun2013dynkin, de2022dynkin}. In our case, by randomizing her trading actions, investor 1 may better conceal the true value of $\mu$, cementing her information advantage, while knowing that excessive randomization breeds highly suboptimal actions.

To manage this tradeoff, investor 1 considers an {\it entropy-regularized} mean-variance objective. Specifically, she chooses a stochastic policy $\Pi=\{\Pi_t\}_{t\in[0,T]}$, a process taking values in the set $\cP(\R)$ of probability measures on $\R$. To generate a randomized trading action at time $t$, she simply samples $u_1(t)\in\R$ from $\Pi_t\in\cP(\R)$. Investor 1's objective (i.e.,  \eqref{J_1^D} below) combines a mean-variance criterion (applied to her mixed wealth $X_1(T)-\lambda_1\overline X(T)$) with an entropy term that quantifies the randomness induced by $\Pi$. That is, investor 1 aims to balance reward achievement and information concealment 
(through randomizing actions via $\Pi$). This is reminiscent of ``exploratory stochastic control'' recently developed for continuous-time reinforcement learning (RL); see e.g., \cite{wang2020reinforcement,wang2020continuous,dai2023learning}. The interpretation is yet different: there, the entropy represents the degree of exploration in RL; here, it reflects how intensely private information is guarded. 

While investor 1, a continuous-time trader by default, may wish to randomize actions continuously over time, it is not theoretically possible due to measurability issues; see the discussion below Remark~\ref{rem:no distinction}. As investor 1 can randomize her actions {\it only} on a discrete time grid $\D$ of $[0,T]$, however dense she likes it to be, her objective (i.e., \eqref{J_1^D}) is $\D$-dependent: while her mixed wealth $X_1-\lambda_1\overline X$ evolves continuously, her trading actions are updated only discretely on $\D$. 

To find a Stackelberg equilibrium between the two investors, we start with solving the follower's (investor 2's) problem. 
Note that what investor 2 observes is a path $\alpha$ of investor 1's actual trading actions, induced jointly by $(\Pi,\D)$ (which investor 2 does not know). Critically, when investor 2 watches the action path $\alpha$ unfold, he has no idea of its dynamics---which again depends on $(\Pi,\D)$. To investor 2, there is then no way of evaluating a usual mean-variance criterion applied to $X_2(T)-\lambda_2\overline X(T)$, his mixed terminal wealth, as the evaluation requires the dynamics of $(X_1,X_2)$, which in turn depends on the dynamics of investor 1's random trading actions. 

To circumvent this, we propose {\it pathwise} optimization for investor 2. As investor 1 updates trading actions discretely, a realized action path $\alpha$ must lie in $\operatorname{Step}([0,T];\R)$, the set of step functions from $[0,T]$ to $\R$. If we {\it freeze} the random trading actions to be a specific path $\alpha\in\operatorname{Step}([0,T];\R)$, investor 1's wealth $X_1$ will become adapted to {\it solely} the stock filtration $\cF^S_\cdot$ (i.e., $X^{(\alpha)}_1$ in \eqref{dynamics investor 2^alpha}), independent of any randomness used for action sampling. Similarly, investor 2's trading strategy $u_2$, which in general reacts to observed trading of investor 1, will become $\cF^S_\cdot$-adapted (see Remark~\ref{rem:G-adapted} for details), hence so will his wealth $X_2$ (see  $X^{(\alpha)}_2$ in \eqref{dynamics investor 2^alpha}). Thus, for each fixed $\alpha\in\operatorname{Step}([0,T];\R)$, investor 2 takes up a mean-variance criterion applied to $X_2^{(\alpha)}(T)-\lambda_2\overline X^{(\alpha)}(T)$, which he can genuinely evaluate based on observed stock evolution; see \eqref{J_2^alpha} for details. 

As a mean-variance criterion induces time inconsistency, investor 2 intends to find an intrapersonal equilibrium $\hat u_2$---a strategy none of his current and future selves will deviate from unilaterally, defined in \eqref{intra E relation} below. While our definition looks similar to the standard one (see e.g., Bj\"{o}rk et al.\ \cite{bjork2017time}), it is in fact {\it stronger} and {\it changes} how we understand ``no deviation.'' In the current setup of a Stackelberg game, our definition requires that 
no matter which action path $\alpha\in\operatorname{Step}([0,T];\R)$ investor 1 takes, ``no deviation'' among investor 2's current and future selves must stand still. Moreover, such ``no deviation'' is revealed or confirmed  {\it in retrospect}, which contrasts the forward-looking perspective in the standard definition; see the detailed discussion below  \eqref{intra E relation}. 

To find investor 2's intrapersonal equilibrium, we derive a {\it pathwise} extended Hamilton--Jacobi--Bellman (HJB) equation (i.e., \eqref{3-5}-\eqref{3-6} below): for each fixed action path $\alpha\in\operatorname{Step}([0,T];\R)$ of investor 1, there is a corresponding extended HJB equation (in the sense of \cite{bjork2017time}) in the form of two coupled partial differential equations (PDEs). Lemma~\ref{lem:hat u} constructs explicitly a solution to these coupled systems, expressed in terms of two Cauchy problems (Lemma~\ref{lem:a_2 bdd}), from which we build a specific intrapersonal equilibrium $\hat u_2$ of investor 2; see Theorem~\ref{thm:follower}. There are three concrete terms in $\hat u_2$: the first two represent ``myopic trading'' and ``intertemporal hedging'' under partial information, 
as recently recognized in \cite{huang2026partial}; the third term, new to the literature, is investor 2's ``information learning''---he adds observed trading of investor 1 (who knows the stock return $\mu$) to his portfolio, up to a constant multiple. See the discussion below Theorem~\ref{thm:follower} for details. 

Note that our pathwise optimization is reminiscent of, but different from, that in \cite{buckdahn2007pathwise}. The agent in \cite{buckdahn2007pathwise} knows the dynamics of an extra process, 
so that he can optimize his objective as a random field. Investor 2 in our case, by contrast, does not know the dynamics of the extra process (i.e., investor 1's action path $\alpha$)---hence no way of optimizing at the level of a random field. Investor 2 instead optimizes the value function under each a priori fixed path $\alpha$; see Remark~\ref{rem:pathwise}. 

We now turn to the leader's (investor 1's) problem. Anticipating investor 2's use of the strategy $\hat u_2$, investor 1 looks for her intrapersonal equilibrium $\Pi$---a stochastic policy none of her current and future selves will deviate from unilaterally. How this should be defined is nontrivial. 
The literature defines an intrapersonal equilibrium in either discrete time or continuous time, yet investor 1 is caught in between: 
she can take a time grid $\D$ (on which she randomizes actions) however dense she likes, but never truly achieves continuous-time trading. Our definition in \eqref{intra E investor 1} below, new to the literature, spells out a recipe for such {\it nearly} continuous controls. In brief, \eqref{intra E investor 1} states that
\begin{equation}\label{The goal}
(\hbox{value difference})^{\Pi,\D}\ \ \le\ \ C|\D|^{1+\kappa}, 
\end{equation}
where the left-hand side is the zeroth-order value difference (with and without deviation from $\Pi$ on $\D$), $|\D|$ is the mesh size of $\D$, and $C, \kappa>0$ are fixed constants. Two key consequences follow. First, under \eqref{The goal}, investor 1 can mitigate deviation evenly on the entire state space (hence achieve an ``$\eps$-intrapersonal equilibrium'') via a single $\D$ dense enough; see Remark~\ref{rem:eps E}. Second, while the use of zeroth-order value difference is similar to a discrete-time definition, when investor 1 lets $\D$ be increasingly dense (i.e., $|\D|\to 0$), due to the exponent $1+\kappa>1$, \eqref{The goal} conforms to a 
first-order value inequality that underlies a {continuous-time} definition; see discussions below Remark~\ref{rem:eps E}. 


To show that investor 1's intrapersonal equilibrium (in the sense of \eqref{The goal}) exists, we will leverage the aforementioned continuous-time RL  framework. By  taking $|\D|\to 0$ in investor 1's $\D$-dependent wealth process and objective, we formally derive an {\it exploratory} stochastic control problem (following \cite{dai2023learning}). Although idealized, this limiting case of continuous action sampling allows explicit analysis: Theorem~\ref{thm:leader}  obtains an exploratory intrapersonal equilibrium $\hat\Pi$ (Definition~\ref{def2.5}, analogous to \cite{dai2023learning}), which is a Gaussian policy explicitly given in \eqref{hat Pi} below. As discussed below Theorem~\ref{thm:leader}, 
the Gaussian mean synthesizes investor 1's knowledge of $\mu$ and internalization of investor 2's reaction $\hat u_2$, while the Gaussian variance reflects how eager investor 1 is to conceal her knowledge of $\mu$. 

In fact,  Lemma~\ref{lem:hat Pi} reveals that $\hat\Pi$ satisfies not only the first-order value inequality that defines an exploratory intrapersonal equilibrium (i.e., \eqref{intra E investor 1 idealized} below), but a {stronger} zeroth-order relation: 
\begin{equation}\label{The goal'}
(\hbox{exploratory value difference})^{\hat\Pi,h}\ \ \le\ \ C_1(\fm^\nu, \fs^\nu) h^{3/2}, 
\end{equation}
where the left-hand side is the zeroth-order value difference (with and without deviation from $\hat\Pi$ on a time period $h>0$) in the exploratory setting and $C_1 >0$ is a constant depending on only the mean $\fm^\nu$ and standard deviation $\fs^\nu$ induced by $\nu\in\cP(\R)$, a deviation from $\hat\Pi$ on the time period $h>0$; see \eqref{J_1(h)-J_1} for details. In view of \eqref{The goal'}, 
it is natural to ask: can $\hat \Pi$, the exploratory intrapersonal equilibrium we find, actually fulfill the desired relation \eqref{The goal}?

Through delicate estimations (requiring Malliavin calculus), Lemma~\ref{lem:(J-J)-(J-J)} identifies that 
the left-hand side of \eqref{The goal} (with $\Pi = \hat \Pi$) and that of \eqref{The goal'} (with $h=|\D|$) 
cannot differ by more than $C_2(\fm^\nu, \fs^\nu, \hat\Pi) |\D|^{3/2}$, where $C_2>0$ depends on only $\fm^\nu$, $\fs^\nu$, and $\hat\Pi$. This, along with \eqref{The goal'},  implies $(\hbox{value difference})^{\hat\Pi,\D}\le C_3(\fm^\nu, \fs^\nu, \hat\Pi) |\D|^{3/2}$, with $C_3(\fm^\nu, \fs^\nu, \hat\Pi):= C_1(\fm^\nu, \fs^\nu)+C_2(\fm^\nu, \fs^\nu, \hat\Pi)$. 
If we further require any deviation $\nu\in\cP(\R)$ to lie within a $K$-Wasserstein ball around $\hat\Pi$ (for any fixed $K>0$), Theorem~\ref{thm:final} finds that $C_3(\fm^\nu, \fs^\nu, \hat\Pi)$ will be uniformly bounded by some $C_K>0$---namely, $\hat\Pi$ fulfills \eqref{The goal} with $C=C_K$ and $\kappa=1/2$. 
We thus call $\hat\Pi$ ``investor 1's intrapersonal equilibrium of order $K$'' (Definition~\ref{def:intra E investor 1}) to stress the dependence on $K>0$, which dictates the realm of allowed deviations from $\hat\Pi$ and can be viewed as part of investor 1's preference (Remark~\ref{rem:preference}). 

The pair $(\hat\Pi, \hat u_2)$ then forms a Stackelberg equilibrium between the investors. With only partial information of $\mu$, investor 2 applies his intrapersonal equilibrium $\hat u_2$, which reacts to real-time trading of investor 1 (who knows the true value of $\mu$). Anticipating investor 2's use of $\hat u_2$, investor 1 employs $\hat\Pi$, her intrapersonal equilibrium of order $K$, to generate random trading actions (to be observed by investor 2). As $\hat\Pi$ satisfies \eqref{The goal}, she can choose a time grid $\D$ (on which she randomizes actions) dense enough to reduce the benefit of deviation as much as she wants.\footnote{As elaborated in \cite{huang2022time}, in a dynamic game where players have time-inconsistent preferences, two levels of game-theoretic thinking are intertwined: when reacting to other players' strategies (the interpersonal level), each player needs to additionally resolve time inconsistency among his current and future selves (the intrapersonal level). At the interpersonal level, \cite{huang2022time, huang2026partial} assume a Cournot structure where players act simultaneously, while we take a Stackelberg structure where a leader moves ahead of a follower (to accommodate asymmetric information).} 

The rest of the paper is organized as follows. Section~\ref{setup} sets up our Stackelberg game, including investor 2's pathwise optimization (Section~\ref{subsubsec:follower}) and investor 1' entropy-regularized problem (Section~\ref{subsubsec:leader}). Section~\ref{sec:follower} derives an intrapersonal equilibrium of investor 2 that reacts to observed trading of investor 1. Given investor 2's reaction, we first study an exploratory version of investor 1's problem (Section~\ref{subsec:exploratory}) and then obtain, via careful approximations, an intrapersonal equilibrium of investor 1 under (realistic) discrete action randomization (Section~\ref{subsec:order K}).

\section{The Setup}\label{setup}
Let $(\Omega,\mathcal{F},\mathbb{P})$ be a probability space 
that supports a standard Brownian motion $W=\{W(s)\}_{s\geq 0}$ and a random variable $\mu:\Omega\to\mathbb R$  independent of $W$. Consider a financial market with a risk-free rate $r>0$ and a stock price process $S=\{S(s)\}_{s\geq 0}$ given by
\begin{align}\label{eq2.1}
    dS(s) &= \mu S(s)\,ds + \sigma S(s)\,dW(s), 
    \qquad S(0)=s_0>0,
\end{align}
where $\sigma>0$ is a given constant. Let $\mathbb{F}^S =\{\mathcal F_s^S\}_{s\geq 0}$ be the filtration generated by $S$ and $\mathbb{F}^{\mu,S}=\{\mathcal{F}_s^{\mu,S}\}_{s\geq 0}$ be the filtration generated by both $\mu$ and $S$ (i.e., $\mathcal{F}_s^{\mu,S} := \mathcal F^S_s \vee \sigma(\mu)$, where $\sigma(\mu)$ is the $\sigma$-algebra generated by $\mu$). 
On a finite time horizon \(T>0\), two investors trade the stock \(S\). The first one, the \textit{informed player}, observes the realization of both $\mu$ and \(S\), hence has access to \(\mathbb F^{\mu,S}\). 
The second investor, the \textit{uninformed player}, observes $S$ without knowing the value of $\mu$, and thus has access to only \(\mathbb F^S\). We assume that $\mu$, the expected return of $S$, has two possible values \(\mu_1\) and \(\mu_2\) (with \(\mu_1>\mu_2\)), and the uninformed player does not know which one is the true value.

For $i=1,2$, let \(u_i(s)\) be investor $i$'s discounted dollar amount invested in $S$ at time \(s\). The discounted wealth process\footnote{
The use of {\it discounted} trading strategies and wealth processes will be convenient for subsequent mean-variance analysis. 
This is motivated by the discounted formulation in \cite{wang2020continuous} (see equation (2) therein).
} \(X_1\) of investor 1 (the informed player) is then
\begin{align}\label{2-1}
	dX_1(s) = (\mu-r)u_1(s)\,ds + \sigma u_1(s)\,dW(s), \quad X_1(0)=x_1 \in \mathbb{R}.
\end{align}
While the discounted wealth process of investor 2 (the uninformed player) similarly satisfies $dX_2(s) = (\mu-r) u_2(s)\,ds + \sigma u_2(s)\,dW(s)$, this dynamics is not useful to him as he is unaware of the true value of \(\mu\). Investor 2, in response, considers the posterior probability
\begin{align*}
	\mathfrak{p}(s):=\mathbb{P}(\mu=\mu_1\mid\mathcal{F}^S_s),\quad s\ge 0. 
\end{align*}
By \cite[Lemma 2]{huang2026partial}, \(\mathfrak p(\cdot)\) can be characterized as the unique strong solution to 
\begin{align}\label{P}
	dP(s)=\frac{\mu_1-\mu_2}{\sigma}P(s)(1-P(s))\,d\widehat{W}(s),\quad s\geq 0,\quad P(0)=\mathfrak p(0)=p\in(0,1),
\end{align} 
where 
\begin{align}\label{eq2.2}
	\widehat{W}(s):= \frac{1}{\sigma} \bigg[\log\bigg(\frac{S(s)}{S(0)}\bigg) -(\mu_1-\mu_2)\int_{0}^s\mathfrak p(\bar s) d\bar s - \bigg(\mu^2-\frac{\sigma^2}{2}\bigg)s \bigg]
	\qquad s\geq 0,
\end{align}
is an \(\mathbb{F}^S\)-adapted standard Brownian motion. By \eqref{eq2.2}, \(S\) in \eqref{eq2.1} can be expressed equivalently as 
\begin{align}\label{eq2.5}
	dS(s) = \big((\mu_1-\mu_2)P(s)+\mu_2\big)S(s)\,ds+\sigma S(s)\,d\widehat{W}(s),\quad s\geq 0,\quad S(0)=s_0>0.
\end{align}
As a result, investor 2's discounted wealth process \(X_2\) can be equivalently expressed as 
\begin{align}\label{2-2}
	dX_2(s) =\big((\mu_1-\mu_2)P(s)+\mu_2-r\big)u_2(s)\,ds +\sigma u_2(s)\,d\widehat{W}(s),\quad X_2(0)=x_2 \in \mathbb{R}.
\end{align}
This dynamics is now fully observable to investor 2, as it involves the known constants \(\mu_1,\mu_2,\sigma\) and $\mathbb F^S$-adapted process \(P(\cdot)\) and $\widehat W(\cdot)$.


\subsection{Randomization of Actions}
In dynamic games with asymmetric information, the informed player may {\it randomize} her control actions to conceal her private information from the uninformed player (who observes the informed player's control actions); see e.g., \cite{cardaliaguet2007differential, cardaliaguet2009stochastic} for two-player zero-sum differential games and \cite{grun2013dynkin, de2022dynkin} for Dynkin games.
In our case, we let $\cP(\R)$ be the set of probability measures on $(\R,\cB(\R))$, where $\cB(\R)$ is the Borel $\sigma$-algebra of $\R$, 
and allow investor 1 to randomize her trading actions, \(u_1(s)\) in \eqref{2-1}, via a stochastic policy defined below. 

\begin{definition}\label{def:Pi}
We say \(\Pi=\{\Pi_s\}_{0\le s\le T}\) is a stochastic policy if \(\Pi_s\in\mathcal P(\mathbb R)\) $\P$-a.s.\ for all $s\in[0,T]$, $\Pi:[0,T]\times\Omega \to \mathcal P(\mathbb R)$ is $\mathbb F^{\mu,S}$-progressively measurable (i.e., for every $t\in [0,T]$, the mapping $(s,\omega)\mapsto \Pi_s(\omega)$ from $([0,t]\times\Omega, \cB([0,t])\otimes\cF_t^{\mu,S})$ to $(\cP(\R),\cB(\cP(\R)))$ is measurable, where $\cB(\cP(\R))$ is the Borel $\sigma$-algebra of $\cP(\R)$), and $\Pi_s$ has a density function for all $s\in[0,T]$.
\end{definition}

\begin{remark}\label{rem:no distinction}
We will not distinguish between $\nu\in\cP(\R)$ and its density function (if it exists). That is, we will abuse the notation slightly by writing $a\mapsto \nu(a)$ for the density function of  $\nu\in\cP(\R)$.  
\end{remark}

To randomize her action at time $s\in [0,T]$, investor 1 simply samples $u_1(s)\in\R$ from $\Pi_s\in\mathcal P(\R)$. While investor 1, a continuous-time trader, might wish to randomize actions continuously over time, it is not theoretically viable. By  \cite[Remark 2.1]{szpruch2024optimal}, for any non-constant and pairwise independent random variables \(\{\xi_s\}_{s\in[0,T]}\), the mapping \(s\mapsto \xi_s\) in general is {\it not} Lebesgue measurable. Thus, if we sample $u_1(s)$ continuously for all $s\in[0,T]$ (via \(\{\xi_s\}_{s\in[0,T]}\)), the path $s\mapsto u_1(s)$ may fail to be progressively measurable, so that \eqref{2-1} may not be well-defined; 
see \cite{bender2024grid} for further discussions.

Hence, investor 1 can randomize her actions $u_1(s)$ in \eqref{2-1} {\it only} on a discrete time grid of $[0,T]$ (however dense it may be). To state this clearly, we follow Jia et al.\ \cite{jia2025accuracy} to introduce \((\Omega^\xi,\mathcal F^\xi,\mathbb P^\xi,\xi)\), where $\xi$ is a random variable defined on the probability space \((\Omega^\xi,\mathcal F^\xi,\mathbb P^\xi)\). Set \(\mathbb N_0:=\mathbb N\cup\{0\}\) and let $\big(\Omega^{\xi_i},\mathcal F^{\xi_i},\mathbb P^{\xi_i},\xi_i\big)_{i\in\mathbb N_0}$ be independent copies of \((\Omega^\xi,\mathcal F^\xi,\mathbb P^\xi,\xi)\). Consider 
\begin{align}\label{space}
(\widetilde{\Omega},\widetilde{\mathcal F},\widetilde{\mathbb P}) := \left(\Omega \times\bar\Omega, \mathcal F \otimes\bar{\mathcal F}, \P\otimes \bar\P\right),\quad\hbox{with}\quad \bar\Omega :=  \prod_{i\in\N_0} \Omega^{\xi_i},\ \bar{\mathcal F} :=\bigotimes_{i\in\N_0}\mathcal F^{\xi_i},\ \bar\P:=  \bigotimes_{i\in\N_0}\mathbb P^{\xi_i}
\end{align}
and denote by \(\widetilde{\mathbb E}\) and \(\widetilde{\operatorname{Var}}\) 
the expectation and variance operators, respectively, under \(\widetilde{\mathbb P}\). 

Given any stochastic policy $\Pi$ and time grid \(\mathcal D=\{t_i\}_{i=0}^n\) of \([0,T]\) (i.e., $0=t_0<t_1<...<t_n=T$), the independent random variables $\{\xi_i\}_{i=0}^{n-1}$ (each defined on $(\Omega^{\xi_i},\mathcal F^{\xi_i},\mathbb P^{\xi_i})$) are used to sample an action $u_1(t_i)\in\R$ from $\Pi_{t_i}\in\mathcal P(\R)$ at time $t_i$, for $i=0,1,...,n-1$. 
Specifically, by \cite[Remark 3.1]{jia2025accuracy}, for any $\nu\in\cP(\R)$, there is a Borel $\phi_\nu:\R\to\R$ 
such that the law of $\phi_\nu(\xi)$ in \((\Omega^\xi,\mathcal F^\xi,\mathbb P^\xi)\) coincides with $\nu$. Hence, under $(\Pi,\mathcal D)$, investor 1's sampled actions take the form $u_1(t_i) = \phi_{\Pi_{t_i}}(\xi_i)$ for $i=0,1,...,n-1$. 
Her wealth dynamics then changes from \eqref{2-1} to 
\begin{align}\label{eq2-14}
dX_1(s) = (\mu-r) u_1(\delta(s))\, ds + \sigma u_1(\delta(s))\,dW(s), \qquad X_1(0)=x_1\in\mathbb R,
\end{align}
where $\delta:[0,T]\to\R$ is defined by \(\delta\equiv t_i\) on \([t_i,t_{i+1})\), for \(i=0,\ldots,n-1\), and \(\delta(T):=t_{n-1}\). 

\begin{remark}\label{rem:Step}
As \(u_1(t_i)\) is randomly sampled at time \(t_i\) and applied on the interval \([t_i,t_{i+1})\) for $i=0,1,.,,,n-1$, for any fixed $\bar\omega\in\bar\Omega$, the realized action path $s\mapsto u_1(\delta(s))(\bar\omega)$ lies in 
\[
\operatorname{Step}([0,T];\R) := \{f:[0,T]\to\R \mid \hbox{$f$ is a right-continuous step function}\}.
\]
Moreover, $\mathsf m:= \bar\P\circ u_1(\delta(\cdot))^{-1}$ defines a probability measure on $(\operatorname{Step}([0,T];\R),\mathcal H_T)$, with 
\begin{equation}\label{H filtration}
\mathcal H_s := \sigma\big(u_1(\delta(t)):t\in [0,s]\big)\ \ \hbox{for}\ s\in [0, T]. 
\end{equation}
\end{remark}



\subsection{Objectives of the Two Investors}
To integrate asymmetric information into the 
portfolio selection game in \cite{huang2026partial}, we use Stackelberg's leader-follower structure: investor 2 (the ``follower'') observes investor 1's trading actions ($u_1(\delta(s))$ in \eqref{eq2-14}) and decides his trading strategy $u_2:[0,T]\times\widetilde\Omega\to\R$ in \eqref{2-2}, which may depend on observed $u_1(\delta(s))$; anticipating investor 2's reaction, investor 1 (the ``leader'') chooses a stochastic policy $\Pi$ and a time grid $\D$, which jointly induce $u_1(\delta(s))\in\R$, to best achieve her objective. Note that investor 2 observes {\it only} actual trading actions $u_1(\delta(s))\in\R$ of investor 1, without knowing the underlying $(\Pi,\mathcal D)$. 
The two investors' choices of strategies are connected
under relative performance concerns. In line with \cite{espinosa2015optimal, lacker2019mean, huang2026partial}, investor $i$ (for $i=1,2$) is concerned about not only the individual terminal wealth \(X_i(T)\), but how it compares to the average terminal wealth of both investors $\overline{X}(T):=\big(X_1^{}(T)+X_2(T)\big)/2$, thereby considering the mixed wealth
\[
(1-\lambda_i) X_i(T) + \lambda_i (X_i(T)-\overline X(T)) = X_i(T) - \lambda_i \overline X(T),
\]
where \(\lambda_i\in[0,1)\) represents the degree of investor $i$'s relative performance concern. 

\subsubsection{The Follower's Problem}\label{subsubsec:follower}
As investor 2 observes actual trading actions $\{u_1(\delta(s))\}_{s\in[0,T]}$ of investor 1 in addition to the stock evolution $\{S_s\}_{s\in [0,T]}$, the filtration available to him is 
\begin{equation}\label{G filtration}
\mathbb G = \{\mathcal G_s\}_{s\in [0,T]},\quad\hbox{with}\quad \mathcal G_s := \mathcal F_s^S \vee \mathcal H_s\ \ \hbox{(where $\mathcal H_s$ is defined in \eqref{H filtration})}. 
\end{equation}
Given any current time \(t\in[0,T]\), wealth levels \(\boldsymbol{x}=(x_1,x_2)\in\mathbb{R}^2\), and posterior probability \(\mathfrak p(t) = p\in(0,1)\), investor 2 attempts to evaluate a $\mathbb G$-progressively measurable process  \(u_2:[0,T]\times\widetilde\Omega\to \R\), his trading strategy, through the (conditional) mean-variance criterion
\begin{align}\label{J_2} 
\widetilde{\mathbb E}_{t,\boldsymbol{x},p}\left[X_2(T)-\lambda_2\overline X(T)\,\middle|\,\mathcal H_T\right]-\frac{\gamma_2}{2}\widetilde{\operatorname{Var}}_{t,\boldsymbol{x},p}\left[X_2(T)-\lambda_2\overline X(T)\,\middle|\,\mathcal H_T\right],
\end{align}
where \(\gamma_2>0\) represents investor 2's risk aversion and the subscript $(t,\bm x,p)$ stands for the initial conditions of the processes 
\begin{align}\label{dynamics investor 2}
\begin{cases}
\ dX_1(s)= \big(\theta(P(s))-r\big) u_1(\delta(s)) \,ds+\sigma u_1(\delta(s))\,d\widehat W(s),\qquad X_1(t)=x_1,\\
\ dX_2(s)=\big(\theta(P(s))-r\big)u_2(s)\,ds+\sigma u_2(s)\,d\widehat W(s),\qquad X_2(t)=x_2,\\
\ dP(s)=\beta(P(s))\,d\widehat W(s),\qquad P(t)=p.
\end{cases}
\end{align}
The second and third lines in \eqref{dynamics investor 2} simply restate \eqref{2-2} and \eqref{P} using the functions
\begin{equation}\label{theta}
\theta(p):=(\mu_1-\mu_2)p+\mu_2 \quad\hbox{and}\quad \beta(p):=\frac{\mu_1-\mu_2}{\sigma}p(1-p),\qquad p\in (0,1).
\end{equation}
The first line in \eqref{dynamics investor 2} rewrites $X_1$ in \eqref{eq2-14} using $\widehat{W}$ in \eqref{eq2.2}---as investor 2 does not know $\mu$ precisely, he recasts \eqref{eq2-14} (involving $\mu$) into the first line of \eqref{dynamics investor 2} (involving $\mathbb F^S$-adapted $P$ and $\widehat W$). 

Note that \eqref{J_2} is conditioned on $\mathcal H_T$, not $\mathcal H_t$. While investor 2 can observe the realization of $u_1(\delta(\cdot))$, he does not know its dynamics (which depends on $(\Pi,\D)$ chosen by investor 1). As $u_2$ is $\mathbb G$-adapted and thus depends on $u_1(\delta(\cdot))$, its dynamics is also unknown to investor 2. Hence, $(X_1,X_2)$ in \eqref{dynamics investor 2} is not useful to investor 2. Only when a given path $\alpha = \{u_1(\delta(s))\}_{s\in [0,T]}\in\operatorname{Step}([0,T];\R)$ is known (under the $\sigma$-algebra $\mathcal H_T$) can \eqref{dynamics investor 2} turn into a form  
accessible to investor 2, i.e.,  
\begin{align}\label{dynamics investor 2^alpha}
\begin{cases}
\ dX^{(\alpha)}_1(s)= \big(\theta(P(s))-r\big)\alpha(s)\,ds+\sigma \alpha(s)\,d\widehat W(s),\qquad X_1^{(\alpha)}(t)=x_1,\\
\ dX_2^{(\alpha)}(s)= \big(\theta(P(s))-r\big) u_2^{(\alpha)}(s)\,ds+\sigma u_2^{(\alpha)}(s)\,d\widehat W(s),\qquad X_2^{(\alpha)}(t)=x_2,\\
\ dP(s)=\beta(P(s))\,d\widehat W(s),\qquad P(t)=p,
\end{cases}
\end{align}
where $u_2^{(\alpha)}$ denotes $u_2$ conditioned on $u_1(\delta(\cdot)) = \alpha\in\operatorname{Step}([0,T];\R)$ (see Remark~\ref{rem:G-adapted} below). 

\begin{remark}\label{rem:G-adapted}
Thanks to \eqref{G filtration}, any $\mathbb G$-adapted $u:[0,T]\times \widetilde\Omega\to\R$ can be expressed as 
$u(s) = f^u_s\left(\{S_\ell\}_{\ell\in[0,s]}, \{u_1(\delta(\ell))\}_{\ell\in[0,s]}\right)$, where $f^u_s: C([0,s];\R)\times \operatorname{Step}([0,T];\R)\to\R$ is measurable for each $s\in [0,T]$. 
Hence, given $\alpha\in \operatorname{Step}([0,T];\R)$, we can define an $\mathbb F^S$-adapted process by 
\[
u^{(\alpha)}(s) := f^u_s\left(\{S_\ell\}_{\ell\in[0,s]}, \{\alpha(\ell)\}_{\ell\in[0,s]}\right),\quad \forall s\in[0,T]. 
\]
If $u$ is adapted to $\mathbb F^S\subset \mathbb G$, then $u^{(\alpha)} = u$ for all $\alpha\in \operatorname{Step}([0,T];\R)$. 
\end{remark}

To ensure that the system \eqref{dynamics investor 2^alpha} is well-defined, we impose suitable admissibility on $u_2$. 

\begin{definition}\label{def:admissible strategies}
A $\mathbb G$-progressively measurable $u_2:[0,T]\times\widetilde\Omega\to\R$ is admissible to investor 2 if $\int_0^T |u^{(\alpha)}_2(s)|^2\, ds<\infty$ $\P$-a.s.\ $\forall \alpha\in \operatorname{Step}([0,T];\R)$. 
We denote by $\mathcal A_2$ the set of all such strategies.
\end{definition}

\begin{remark}\label{rem:P and X_1}
Fix any $(t,\bm x,p)\in[0,T]\times\R^2\times (0,1)$ and $\alpha\in \operatorname{Step}([0,T];\R)$. By \cite[Lemma 3.2]{huang2026partial}, ``$dP(s)=\beta(P(s))\,d\widehat W(s)$, $P(t)=p$'' in \eqref{dynamics investor 2^alpha} has a unique strong solution, which lies in $(0,1)$ $\P$-a.s. Hence, in view of \eqref{theta}, $\theta(P(s))$ is a bounded process. This, along with $\alpha(s)$ being a step function, ensures that $X_1^{(\alpha)}$ in \eqref{dynamics investor 2^alpha} is well-defined. Finally, for any $u_2\in\A_2$, the admissibility condition in Definition~\ref{def:admissible strategies} and boundedness of $\theta(P(s))$ imply that $X_2^{(\alpha)}$ in \eqref{dynamics investor 2^alpha} is well-defined. 
\end{remark}

For each $(t,\bm x,p)$, the objective \eqref{J_2} is an \(\mathcal H_T\)-measurable random variable on $(\widetilde\Omega,\widetilde\cF,\widetilde\P)$. One may thus, in principle, view it as a random field (instead of a real-valued function) and optimize it in the essential sense (e.g., achieving essential supremums in the sense of \cite[Appendix A]{karatzas1998methods}). This, however, is infeasible to investor 2. By recalling the definition of conditional expectations and using \eqref{dynamics investor 2^alpha}, we may write \eqref{J_2} in terms of 
a random variable $Y$ on $(\widetilde\Omega,\widetilde\cF,\widetilde\P)$ defined by
\[
Y(\omega,\bar\omega) := \left(X^{(u_1(\delta(\cdot))(\bar\omega))}_2(T)\right)(\omega) -\lambda_2 \left(\overline{X}^{(u_1(\delta(\cdot))(\bar\omega))}(T)\right)(\omega),\quad\forall (\omega,\bar\omega)\in\widetilde\Omega, 
\]
where we set $\overline{X}^{(\alpha)} := ({X}^{(\alpha)}_1 + {X}^{(\alpha)}_2)/2$. 
Note that the law of $Y$ depends on the law of $u_1(\delta(\cdot))$ (i.e., $\mathsf m$ in Remark~\ref{rem:Step}), which in turn depends on $(\Pi,\D)$---unknown to investor 2. Without knowing the law of $Y$, investor 2 is unable to optimize \eqref{J_2} in the essential sense. 

The only option left to investor 2 is {\it pointwise} optimization. For any fixed path $\alpha\in\operatorname{Step}([0,T];\R)$ taken by  $u_1(\delta(\cdot))$, investor 2 considers \eqref{dynamics investor 2^alpha} and evaluates the objective
\begin{equation}\label{J_2^alpha}
J^{(\alpha)}_2(t,\bm{x},p;u_2):= {\mathbb E}_{t,\bm{x},p}\left[X_2^{(\alpha)}(T)-\lambda_2\overline X^{(\alpha)}(T)\right]-\frac{\gamma_2}{2}{\operatorname{Var}}_{t,\bm{x},p}\left[X_2^{(\alpha)}(T)-\lambda_2\overline X^{(\alpha)}(T)\right],  
\end{equation}
where we use $\overline{X}^{(\alpha)} := ({X}^{(\alpha)}_1 + {X}^{(\alpha)}_2)/2$ again. With $\alpha\in \operatorname{Step}([0,T];\R)$ fixed, all processes in \eqref{dynamics investor 2^alpha} are defined on $(\Omega,\mathcal F,\P)$ and adapted to $\mathbb F^S$, so that investor 2 can indeed evaluate \eqref{J_2^alpha}. 

\begin{remark}\label{rem:pathwise}
Our pathwise optimization is reminiscent of, but different from, that in \cite{buckdahn2007pathwise}. There, an investor observes some extra noise and also {\it knows} its dynamics. His objective is an expectation conditioned on the $\sigma$-algebra generated by the noise on $[0,T]$, like \eqref{J_2}, and it is optimized as a random field, i.e., a strategy is chosen to achieve the essential supremum of the objective. In our case, as just explained, 
without knowing the dynamics of $u_1(\delta(\cdot))$, investor 2 is unable to optimize \eqref{J_2} as a random field. He can only a priori fix a path $\alpha$ of $u_1(\delta(\cdot))$ and optimize the resulting value function \eqref{J_2^alpha}. This shifts the usual forward-looking perspective in stochastic control to a backward-confirmation one; see the discussion below Definition~\ref{def:intra E investor 2} for specifics. 
\end{remark}

As a mean-variance objective induces time inconsistency, investor 2 aims to find an intrapersonal equilibrium $u_2$---a strategy none of his current and future selves will deviate from unilaterally---
under \eqref{J_2^alpha}. To define such a strategy precisely, we first introduce Markov strategies. 

\begin{definition}\label{def:Markov strategies}
We say $u_2\in\A_2$ is Markov if there exists a Borel $\eta:[0,T]\times\R^2\times (0,1)\times \R\to\R$ such that for any $(t,\bm x,p)\in [0,T]\times\R^2\times (0,1)$ and $\alpha\in \operatorname{Step}([0,T];\R)$, 
\begin{equation}\label{Markov}
u_2^{(\alpha)}(s) = \eta\Big(s,X_1^{(\alpha)}(s),X_2^{(\alpha)}(s),P(s),\alpha(s)\Big)\ \ \hbox{for a.e.}\ s\in[t,T],\quad \hbox{$\P$-a.s.},
\end{equation}
with $(X_1^{(\alpha)}, X_2^{(\alpha)}, P)$ as in \eqref{dynamics investor 2^alpha}. 
We will abuse the notation slightly by writing $u_2$ and $\eta$ interchangeably (e.g, ``$\eta\in\A_2$ is Markov'' indicates an underlying $u_2\in\A_2$ for which \eqref{Markov} holds). 
\end{definition}


\begin{definition}
\label{def:intra E investor 2}
For any Markov $u_2\in\A_2$, \(t\in[0,T)\), \(h>0\) with \(t+h\le T\), and $v\in\R$, we define
\[
u_2^{h,v}(s) :=
\begin{cases}
v, & s\in[t,t+h),\\
u_2(s), & s\in[t+h,T].
\end{cases}
\]
We say \(\hat u_2\in\mathcal A_2\) is investor 2's intrapersonal equilibrium if $\hat u_2$ is Markov and for every \((t,\bm{x},p)\in[0,T)\times\R^2\times (0,1)\), $\alpha\in \operatorname{Step}([0,T];\R)$, and $v\in\R$, 
\begin{align}\label{intra E relation} 
\limsup_{h\downarrow0}
\frac{J_2^{(\alpha)}(t,\boldsymbol{x},p;\hat u_2^{h,v}) - J_2^{(\alpha)}(t,\boldsymbol{x},p;\hat u_2)}{h} \le 0. 
\end{align}
\end{definition}

Condition \eqref{intra E relation} seemingly states: ``at any time and state $(t,\bm x,p)$, as long as investor 2's future selves will employ $\hat u_2$, his current self cannot be better off by deviating to other actions $v$''---the same message from the standard definition of an intrapersonal equilibrium for one single agent (see e.g., Bj\"{o}rk et al.\ \cite{bjork2017time}). In fact, what \eqref{intra E relation} conveys, in the context of a Stackelberg game,  is not only {\it stronger} but {\it changing} 
how we understand ``no deviation.'' 
First, \eqref{intra E relation} stipulates ``no deviation'' under {\it any} behavior of the other player (i.e., {\it any} action path $\alpha$ investor 1 takes). 
Second, such ``no deviation'' is revealed or confirmed  {\it in retrospect}. As the actual action path can be fully observed only at time $T$, 
what  \eqref{intra E relation} truly states is: ``no matter which action path $\alpha$ is eventually realized, it will {\it always} be found (or confirmed) that it was good for investor 2 to follow $\hat u_2$ at any previous time and state $(t,\bm x,p)$, as taking other actions $v$ back then could have made him worse off.''


\subsubsection{The Leader's Problem}\label{subsubsec:leader}
With full information (i.e., $\mathbb F^{\mu,S}$), investor 1 can take any stochastic policy $\Pi$ (Definition~\ref{def:Pi}).
Also, as the leader of the game, investor 1 anticipates investor 2's use of an intrapersonal equilibrium $\hat u_2$ (Definition~\ref{def:intra E investor 2}), which reacts to investor 1's actual trading actions $\{u_1(\delta(s))\}_{s\in [0,T]}$ (induced by $\Pi$ on a time grid $\D$ of $[0,T]$). 
Given any current time \(t\in[0,T]\), wealth levels \(\boldsymbol{x}=(x_1,x_2)\in\mathbb{R}^2\), and posterior probability \(\mathfrak p(t) = p\in(0,1)\), investor 1 evaluates a policy $\Pi$ over a time grid $\D=\{t_i\}_{i=0}^n$ through the (entropy-regularized) mean-variance objective
\begin{align}\label{J_1^D}
J_1^{\D}(t,\boldsymbol x,p;\Pi) &:= \widetilde{\mathbb E}_{t,\boldsymbol{x},p}\bigg[X_1(T)-\lambda_1\overline X(T)+\lambda_0\sum_{i:\,t_i\in[t,T)}H(\Pi_{t_i})(t_{i+1}-t_i)\bigg] \notag\\
&\qquad-\frac{\gamma_1}{2}\widetilde{\operatorname{Var}}_{t,\boldsymbol{x},p}\left[X_1(T)-\lambda_1\overline X(T)\right],
\end{align}
where $H(\nu):= -\int_{\mathbb R}\nu(a)\log\nu(a)\,da$ is the differential entropy of $\nu\in\mathcal P(\R)$, $\lambda_0>0$ is the weight investor 1 places on the entropy term, \(\gamma_1>0\) represents investor 1's risk aversion, and the subscript $(t,\bm x,p)$ stands for the initial conditions of the processes 
\begin{align}\label{dynamics investor 1}
\begin{cases}
\ dX_1(s) = (\mu-r) u_1(\delta(s))\,ds + \sigma u_1(\delta(s))\,dW(s), \qquad X_1(t)=x_1,\\
\ dX_2(s)= (\mu-r) \hat u_2(s)\,ds+\sigma \hat u_2(s)\,dW(s),\qquad X_2(t)=x_2,\\
\ dP(s)=\left[\frac{\mu-\mu_2}{\sigma}-\frac{\mu_1-\mu_2}{\sigma}P(s)\right]\beta(P(s))\,ds+\beta(P(s))\,dW(s),\qquad P(t)=p.
\end{cases}
\end{align}

The entropy term in \eqref{J_1^D} quantifies the randomness of a stochastic policy $\Pi$. That is, investor 1 aims to strike a balance between achieving her reward objective (the mean-variance structure in \eqref{J_1^D}) and concealing the true value of $\mu$ (by randomizing trading actions through $\Pi$). This setup is reminiscent of the recent stochastic control framework for continuous-time RL (see e.g., \cite{wang2020reinforcement,wang2020continuous,dai2023learning}), yet the interpretation is different; see Remark~\ref{rem:different interpretations} for a related discussion. 

In \eqref{dynamics investor 1}, the first line simply restates \eqref{eq2-14}; the second and third lines rewrite \eqref{2-2} and \eqref{P} under $W$ (using \eqref{eq2.2}), respectively. This rewrite under $W$ is possible for investor 1 (as she knows $\mu$) and allows all processes in \eqref{dynamics investor 1} to involve only $W$, instead of both $W$ and $\widehat W$. 
Note that investor 1 considers $P$ only because
she takes into account the strategy $\hat u_2$ of investor 2 (who uses $P$ to estimate $\mu$, so that $\hat u_2$ in general depends on $P$). 

\begin{remark}
The system \eqref{dynamics investor 1} is well-defined by the same arguments in Remark~\ref{rem:P and X_1}.  
\end{remark}  

For \eqref{J_1^D} to be well-defined on any time grid $\D$ (however dense it is), we impose the following. 

\begin{definition}\label{def:admissible strategies'}
A stochastic policy $\Pi$ is admissible to investor 1 if $\mathbb E\big[\int_0^T |H(\Pi_s)|\,ds\big]<\infty$. We denote by $\mathcal A_1$ the set of all such policies. 
\end{definition}

As a mean-variance objective induces time inconsistency, investor 1 intends to find an intra-personal equilibrium $\hat\Pi$---a policy none of her current and future selves will deviate from unilaterally---under \eqref{J_1^D}. To define such a policy precisely, we first introduce Markov policies. 

\begin{definition}
We say $\Pi\in \A_1$ is Markov if for any $(t,\bm x,p)\in [0,T]\times\R^2\times [0,1]$, 
\[
\Pi_s = \pi\big(s,X_1(s),X_2(s),P(s),\cdot\big),\quad s\in [t,T],
\]
where $\pi:[0,T]\times\R^2\times (0,1) \times\cB(\R)\to[0,1]$ is a probability kernel (i.e., $\pi(t,\bm x,p,\cdot)\in\cP(\R)$ for any fixed $(t,\bm x,p)$ and $(t,\bm x,p)\mapsto \pi(t,\bm x,p,A)$ is Borel for any fixed $A$) and $(X_1,X_2,P)$ fulfills \eqref{dynamics investor 1}. 
\end{definition}

We also need to consider $\cP_2(\R)$, the set of all $\nu\in\cP(\R)$ with finite second moments. Each $\nu\in\cP_2(\R)$ has corresponding {\it mean} and {\it standard deviation}, i.e.,
\begin{equation}\label{tb}
\mathfrak m^\nu :=\int_{\mathbb R} a \nu(a)\, da \qquad \hbox{and}\qquad
\mathfrak s^\nu := \left(\int_{\mathbb R} a^2\nu(a)\,da-(\mathfrak m^\nu)^2 \right)^{1/2}, 
\end{equation} 
By the quantile formulation of the Wasserstein-2 metric $\mathcal W_2(\cdot,\cdot)$ on $\cP_2(\R)$ (see e.g., \cite[Chapter 2]{Santambrogio-book-15}), it can be checked that for any $\nu, \nu'\in\cP_2(\R)$, 
\begin{equation}\label{W_2}
\mathcal W_2(\nu, \nu') = \sqrt{(\fm^\nu-\fm^{\nu'})^2 + (\fs^\nu-\fs^{\nu'})^2+2\fs^\nu\fs^{\nu'}(1-\rho^{\nu,\nu'})} 
\end{equation}
where $\rho^{\nu,\nu'} := [\int_{0}^1(F^{-1}_\nu(y)-\fm^\nu) (F^{-1}_{\nu'}(y)-\fm^{\nu'})\, dy]/(\fs^\nu\fs^{\nu'})\in [-1,1]$ is the {\it quantile correlation coefficient} of $\nu$ and $\nu'$ (where $F^{-1}_\nu$ is the quantile function of $\nu$). 
For any $\nu\in\cP_2(\R)$ and $K>0$, we also consider the Wasserstein ball
\[
B_K(\nu) := \{\bar\nu\in\cP_2(\R) : \mathcal W_2(\bar\nu,\nu)<K\}. 
\]

\begin{definition}
\label{def:intra E investor 1}
For any Markov $\Pi\in\A_1$, \(t, h\ge 0\) with \(t+h\le T\), and $\nu\in\cP_2(\R)$, we define 
\begin{equation}\label{Pi^h}
\Pi^{h,\nu}_s =
\begin{cases}
\nu, & s\in[t,t+h),\\
\Pi_s, & s\in[t+h,T].
\end{cases}
\end{equation}
Given an intrapersonal equilibrium $\hat u_2$ of investor 2 (Definition~\ref{def:intra E investor 2}) and $K>0$, $\hat\Pi\in\A_1$ is investor 1's intrapersonal equilibrium of order $K$  (w.r.t.\  $\hat u_2$) if $\hat\Pi$ is Markov and 
for all time grid $\D=\{t_i\}_{i=0}^n$, 
\begin{equation}\label{intra E investor 1}
J_1^\D\big(t_i,\bm x,p;\hat\Pi^{h_i,\nu}\big) - J_1^\D\big(t_i,\bm x,p;\hat\Pi\big) \le C |\D|^{1+\kappa},\quad \forall(t_i,\bm x,p)\in \D\times \R^2 \times (0,1),\ \nu\in B_K(\hat\Pi_{t_i}),
\end{equation}
where $h_i := t_{(i+1)\wedge n}-t_i$, $|\mathcal D|:=\max_{0\le i\le n-1}(t_{i+1}-t_i)$, and $C,\kappa>0$ are constants independent of $(t_i,\bm x,p,\nu)$. The tuple $(\hat \Pi,\hat u_2)$ is called a Stackelberg equilibrium between the two investors.  
\end{definition}

\begin{remark}\label{rem:preference}
A larger $K>0$ in Definition~\ref{def:intra E investor 1} means that investor 1 at time $t_i$ allows a larger realm of deviations from $\hat\Pi_{t_i}$ (i.e., the ball $B_K(\hat\Pi_{t_i})$). In a sense, $K>0$ represents investor 1's attitude  (``conservative'' or ``aggressive'') towards deviation and is part of investor 1's preference. 
\end{remark}

\begin{remark}\label{rem:eps E}
For any $\eps>0$, \eqref{intra E investor 1} ensures that $\hat\Pi$ is an ``$\eps$-intrapersonal equilibrium'' on any time grid $\D$ dense enough. 
Indeed, as long as $\D$ fulfills $|\D|\le (\eps/C)^{1/(1+\kappa)}$, \eqref{intra E investor 1} yields 
\begin{equation}\label{eps intra E}
J_1^\D\big(t_i,\bm x,p;\hat\Pi^{h_i,\nu}\big) - J_1^\D\big(t_i,\bm x,p;\hat\Pi\big) \le \eps,\quad \forall(t_i,\bm x,p,\nu)\in \D\times \R^2 \times (0,1)\times B_K(\hat\Pi_{t_i}).
\end{equation}
\end{remark}

Condition \eqref{intra E investor 1}, new to the time-inconsistent stochastic control literature, captures two somewhat conflicting aspects of investor 1: a {\it continuous-time trader} obliged to take (randomized) actions only {\it discretely}, as explained below Remark~\ref{rem:no distinction}. With actions taken only on a discrete time grid $\D$, the left-hand side of \eqref{intra E investor 1} presents a direct (zeroth-order) comparison between values with and without deviation, just like how a {\it discrete-time} intrapersonal equilibrium is defined. Yet, investor 1, a continuous-time trader by default, can make the time grid $\D$ however dense she wants (to approximate continuous-time trading). As a result, as the mesh size $|\D|$ tends to 0, the zeroth-order value comparison should conform to a derivative-type (first-order) value comparison (like \eqref{intra E investor 1 idealized} below) that defines a {\it continuous-time} intrapersonal equilibrium. This is where the right-hand side of  \eqref{intra E investor 1} comes into play. By dividing both sides of \eqref{intra E investor 1} by $|\D|$, we see that
\[
\limsup_{|\D|\to 0} \frac{J_1^\D\big(t_i,\bm x,p;\hat\Pi^{h_i,\nu}\big) - J_1^\D\big(t_i,\bm x,p;\hat\Pi\big)}{|\D|}\le 0, 
\]
which aligns with the continuous-time definition \eqref{intra E investor 1 idealized}. It is worth noting, however, that \eqref{intra E investor 1} is {\it stronger} than a direct discretization of \eqref{intra E investor 1 idealized}. In \eqref{intra E investor 1 idealized}, as we care about solely the ``continuous-time limit'' as $h\downarrow 0$, the rate of convergence w.r.t.\ $h$ does not matter (and is allowed to vary widely with $t_i$, $\bm x$, $p$, and $\nu$). 
In our case, investor 1 can never achieve continuous-time trading and is always confronted with a fixed time step $|\D|$, however small it is. A {\it uniform} rate of convergence w.r.t.\ $|\D|$ (independent of $t_i$, $\bm x$, $p$, and $\nu$) is then important: it gives rise to \eqref{eps intra E}, allowing investor 1 to uniformly reduce the benefit of deviation through a single $\D$ dense enough. 

The goal of this paper is to find a Stackelberg equilibrium $(\hat\Pi,\hat u_2)$ for the two investors (Definition~\ref{def:intra E investor 1}). We will start with solving the follower's (investor 2's) optimization problem.


\section{The Follower's Optimization Problem}\label{sec:follower}
This section focuses on finding an intrapersonal equilibrium $\hat u_2$ of investor 2 (Definition~\ref{def:intra E investor 2}). The first step is to introduce an equivalent form of \eqref{J_2^alpha} that is technically more convenient. 

For any $(t,\bm x=(x_1,x_2), p)\in [0,T]\times\R^2\times(0,1)$ and $\alpha\in\operatorname{Step}([0,T];\R)$, recall \eqref{dynamics investor 2^alpha} and define
\begin{equation}\label{Z_2 def}
Z_2^{(\alpha)}(s) := X_2^{(\alpha)}(s)-\lambda_2\overline X^{(\alpha)}(s) = \left(1-\frac{\lambda_2}{2}\right)X_2^{(\alpha)}(s) - \frac{\lambda_2}{2}X_1^{(\alpha)}(s) ,\quad s\in [t,T],
\end{equation}
with $\overline X^{(\alpha)}:= (X_1^{(\alpha)}+ X_2^{(\alpha)})/2$. Thanks to \eqref{dynamics investor 2^alpha}, $Z_2^{(\alpha)}$ satisfies
\begin{align}\label{Z_2} 
dZ_2^{(\alpha)}(s) = \big(\theta(P(s))-r\big) u^{(\alpha)}(s)\,ds + \sigma u^{(\alpha)}(s)\,d\widehat W(s),\quad Z_2^{(\alpha)}(t)=z_2, 
\end{align}
with
\begin{equation}\label{z_2}
z_2 = \bigg(1-\frac{\lambda_2}{2}\bigg)x_2 - \frac{\lambda_2}{2}x_1\quad\hbox{and}\quad u(s) = \bigg(1-\frac{\lambda_2}{2}\bigg)u_2(s) - \frac{\lambda_2}{2}u_1(\delta(s))
\end{equation}
(so that by Remark~\ref{rem:G-adapted}, $u^{(\alpha)}(s) = \big(1-\frac{\lambda_2}{2}\big)u_2^{(\alpha)}(s) - \frac{\lambda_2}{2}\alpha(s)$). 
Conversely, for any $(t,z_2, p)\in [0,T]\times\R\times(0,1)$, $\alpha\in\operatorname{Step}([0,T];\R)$, and $\mathbb G$-progressively measurable $u:[0,T]\times\widetilde\Omega\to\R$, we consider $Z_2^{(\alpha)}$ in \eqref{Z_2} (with $u^{(\alpha)}$ given as in Remark~\ref{rem:G-adapted}) and define
\begin{equation}\label{cJ_2^alpha}
\cJ_2^{(\alpha)}(t,z_2,p; u) := {\mathbb E}_{t,z_2,p} \left[Z_2^{(\alpha)}(T)\right]   -  \frac{\gamma_2}{2}{\operatorname{Var}}_{t,z_2,p} \left[Z_2^{(\alpha)}(T)\right], 
\end{equation}
where the subscript $(t,z_2,p)$ denotes conditioning on ``$Z_2^{(\alpha)}(t)=z_2$ and $P(t)=p$.'' Under \eqref{z_2}, we deduce from \eqref{J_2^alpha} and \eqref{cJ_2^alpha} that 
\begin{equation}\label{J_2=cJ_2}
J_2^{(\alpha)}(t,\bm x,p;u_2) = \mathcal J_2^{(\alpha)}(t,z_2,p;u),\quad \forall \alpha \in \operatorname{Step}([0,T],\R).
\end{equation}
Hence, the intrapersonal equilibrium relation \eqref{intra E relation} can be equivalently stated as 
\begin{equation}\label{intra E relation u} 
\limsup_{h\downarrow0}
\frac{ \cJ_2^{(\alpha)}(t,z_2,p;\hat u^{h,v}) - \cJ_2^{(\alpha)}(t,z_2,p;\hat u)}{h} \le 0. 
\end{equation}
To find such a $\hat u$, we will follow Bj\"{o}rk et al.\ \cite{bjork2017time} to derive the associated extended HJB equation. 

What sets us apart from \cite{bjork2017time} (and most studies on time-inconsistent stochastic control) is that the objective \eqref{cJ_2^alpha} additionally depends on $\alpha\in\operatorname{Step}([0,T];\R)$, {\it externally} determined by investor 1. In a sense, investor 2 considers the objective \eqref{cJ_2^alpha} {\it simultaneously} for all $\alpha$---no matter which action path $\alpha$ investor 1 takes, \eqref{intra E relation u} has to be fulfilled. Hence, every fixed $\alpha\in\operatorname{Step}([0,T];\R)$ should correspond to an extended HJB equation, i.e., a system of two coupled PDEs for the related mean-variance problem. Given a candidate $\hat u$ and its value function $\cV_2^{(\alpha)}(t,z_2,p) := \mathcal J_2^{(\alpha)}(t,z_2,p;\hat u)$, for each fixed $\alpha\in\operatorname{Step}([0,T];\R)$, by following the derivation in \cite{bjork2017time} under the dynamics of \((Z_2^{(\alpha)},P)\) in \eqref{Z_2} and \eqref{dynamics investor 2^alpha}, we obtain
\begin{align}\label{3-5}
0=\partial_t \cV^{(\alpha)}_2 + \sup_{\mathfrak u\in\mathbb R}\Big\{&(\theta(p)-r)\mathfrak u\,\partial_{z_2}\cV^{(\alpha)}_2+\frac{1}{2}\sigma^2 \mathfrak u^2\partial_{z_2z_2}\cV_2^{(\alpha)}+\frac{1}{2}\beta^2(p)\partial_{pp}\cV_2^{(\alpha)}+\sigma\beta(p)\mathfrak u\,\partial_{z_2p}\cV_2^{(\alpha)} \notag\\
&-\frac{\gamma_2}{2}\sigma^2\mathfrak u^2(\partial_{z_2}g_2^{(\alpha)})^2 -\frac{\gamma_2}{2}\beta^2(p)(\partial_p g_2^{(\alpha)})^2 -\gamma_2\sigma\beta(p)\mathfrak u\,\partial_p g_2^{(\alpha)}\,\partial_{z_2}g_2^{(\alpha)}\Big\},
\end{align}
with terminal condition \(\cV^{(\alpha)}_2(T,z_2,p)=z_2\), where $g_2^{(\alpha)}(t,z_2,p)$ 
is the solution to
\begin{align}\label{3-6}
0 = \partial_t g_2^{(\alpha)} +(\theta(p)-r)\hat u^{(\alpha)}\,\partial_{z_2}g_2^{(\alpha)} +\frac{\sigma^2}{2} (\hat u^{(\alpha)})^2\partial_{z_2z_2}g_2^{(\alpha)}+\frac{\beta^2(p)}{2}\partial_{pp}g_2^{(\alpha)} +\sigma\beta(p)\hat u^{(\alpha)}\partial_{z_2p}g_2^{(\alpha)},
\end{align}
with terminal condition \(g_2^{(\alpha)}(T,z_2,p)=z_2\). 


Our goal is to find a solution $(\cV_2^{(\alpha)}, g_2^{(\alpha)})$ to \eqref{3-5}-\eqref{3-6} such that the maximizer of the supremum in \eqref{3-5} coincides with $\hat u^{(\alpha)}$, to be achieved for {\it all} $\alpha\in\operatorname{Step}([0,T];\R)$. To possibly simplify the problem, we note that if $\hat u$ does not depend on $u_1(\delta(\cdot))$) (i.e., $\hat u$ is adapted to $\mathbb F^S\subset\mathbb G$), then $\hat u^{(\alpha)}= \hat u$ no longer depends on $\alpha$ (recall Remark~\ref{rem:G-adapted}), which in turn suggests that \eqref{3-5}-\eqref{3-6} may admit a solution $(\cV_2^{(\alpha)},g_2^{(\alpha)})=  (\cV_2, g_2)$ independent of $\alpha$. 
We propose that such a solution exists and takes the following form: for any $(t,z_2,p)\in[0,T]\times\R\times (0,1)$, consider
\begin{align}\label{V_2 and g_2}
\cV_2(t,z_2,p) &:= z_2+A_2(t,p),\qquad g_2(t,z_2,p):= z_2+a_2(t,p),
\end{align}
\begin{equation}\label{hat u}
\hat u(t,p) := \frac{\theta(p)-r}{\gamma_2\sigma^2} - \frac{\beta(p)\partial_pa_2(t,p)}{\sigma}, 
\end{equation}
where $a_2$ is the solution 
to the Cauchy problem
\begin{equation}\label{Cauchy 1}
\begin{cases}
\partial_t a_2 + \frac{(\theta(p)-r)^2}{\sigma^2\gamma_2} - \frac{\beta(p)(\theta(p)-r)\partial_p a_2}{\sigma} + \frac{1}{2}\beta(p)^2\partial_{pp}a_2 =0, \quad & (t,p)\in[0,T)\times(0,1),\\
a_2(T,p)=0,\quad & p\in(0,1),
\end{cases}
\end{equation}
and \(A_2\) is the solution 
to the Cauchy problem
\begin{align}\label{Cauchy 2}
\begin{cases}
\partial_t A_2 + \frac{1}{2}\beta^2(p)\partial_{pp}A_2 + \mathcal R(t,p,\partial_p a_2) =0, & (t,p)\in[0,T)\times(0,1),\\
A_2(T,p)=0, & p\in(0,1),
\end{cases}
\end{align}
with
\begin{align*}
\mathcal R(t,p,\partial_p a_2):={}
(\theta(p)-r)
\left[
\frac{\theta(p)-r}{\sigma^2\gamma_2} - \frac{\beta(p)\partial_p a_2}{\sigma} \right] - \frac{\gamma_2}{2} \left( \sigma
\left[\frac{\theta(p)-r}{\sigma^2\gamma_2}- \frac{\beta(p)\partial_p a_2}{\sigma}\right] + \beta(p)\partial_p a_2\right)^2.
\end{align*}

Thanks to \cite{huang2026partial}, \eqref{Cauchy 1} (resp.\ \eqref{Cauchy 2}) admits a unique classical solution $a_2$ (resp.\ $A_2$). Furthermore, by classical parabolic estimates in \cite{Friedman-book-64} and a recursive use of $L^2$-derivatives of stochastic processes in \cite{friedman1975stochastic}, we find that many of the derivatives of $a_2$ are bounded. All this is presented in the next result, whose proof is relegated to Appendix~\ref{subsec:proof of lem:a_2 bdd}. 

\begin{lemma}\label{lem:a_2 bdd}
\begin{itemize}
\item [(i)] The Cauchy problem \eqref{Cauchy 1} has a unique solution $a_2$ in $C^{1,2}([0,T)\times(0,1))\cap C([0,T]\times(0,1))$. Moreover, (a) $a_2$ is infinitely differentiable w.r.t.\ $p$; (b) $\partial_{p}\partial_t a_2$ and $\partial_{t}\partial_p a_2$ exist and are equal; (c) $a_2$, $\partial_t a_2$, and all the derivates mentioned are bounded on $[0,T]\times (0,1)$.
\item [(ii)] The Cauchy problem \eqref{Cauchy 2} has a unique solution $A_2$ in $C^{1,2}([0,T)\times(0,1))\cap C([0,T]\times(0,1))$. Moreover, $A_2$ and $\partial_p A_2$ are bounded on $[0,T]\times (0,1)$. 
\end{itemize} 
\end{lemma}

\begin{remark}
The boundedness of $a_2$ and its various derivatives in Lemma~\ref{lem:a_2 bdd} will play an important role in subsequent estimates in Lemmas~\ref{lem:hat u} and \ref{lem:(J-J)-(J-J)}. 
\end{remark}

By Lemma~\ref{lem:a_2 bdd}, $(\cV_2, g_2)$ and $\hat u$ in \eqref{V_2 and g_2}-\eqref{hat u} are well-defined. The next result (whose proof is relegated to Appendix~\ref{subsec:proof of lem:hat u}) shows that they solve the pathwise extended HJB equation \eqref{3-5}-\eqref{3-6}, 
which ultimately leads to the zeroth-order estimate for $\cJ_2^{(\alpha)}(t,z_2,p;\hat u)$ in \eqref{J^alpha-J^alpha<=}.

\begin{lemma}\label{lem:hat u}
$(\cV_2, g_2)$ and $\hat u$ in \eqref{V_2 and g_2}-\eqref{hat u} satisfy the following.
\begin{itemize} 
\item [(i)] $\hat u\in\A_2$ is bounded and Markov, with $\hat u^{(\alpha)} = \hat u$ for all $\alpha\in\operatorname{Step}([0,T];\R)$. 
\item [(ii)] $(\cV_2, g_2)$ and $\hat u$ (viewed as constant functions in $\alpha$) satisfy \eqref{3-5}-\eqref{3-6} and the maximizer of the supremum in \eqref{3-5} coincides with $\hat u$.  
\item [(iii)] For any $(t,z_2,p)\in[0,T]\times\R\times (0,1)$, let $\hat Z_2$ satisfy \eqref{Z_2} with $u = \hat u$. Then, 
\begin{align*}
g_2(t,z_2,p) &= \mathbb E_{t,z_2,p}\big[\hat Z_2(T)\big],\\ 
\cV_2(t,z_2,p) &=  {\mathbb E}_{t,z_2,p} \big[\hat Z_2(T)\big]   -  \frac{\gamma_2}{2}{\operatorname{Var}}_{t,z_2,p} \big[\hat Z_2(T)\big]= \cJ_2^{(\alpha)}(t,z_2,p;\hat u),\ \ \forall \alpha\in\operatorname{Step}([0,T];\R). 
\end{align*}
\item [(iv)] For any $(t,z_2,p)\in[0,T]\times\R\times (0,1)$, $\alpha\in\operatorname{Step}([0,T];\R)$, $v\in\R$, and $h>0$, 
\begin{equation}\label{J^alpha-J^alpha<=}
\cJ_2^{(\alpha)}(t,z_2,p;\hat u^{h,v}) - \cJ_2^{(\alpha)}(t,z_2,p;\hat u)\le \mathsf C_0 (1+v^2) \big(h^{3/2} + {h^2}\big),
\end{equation}
where $\mathsf C_0>0$ is independent of $(t,z_2,p)$, $v$, and $h$. 
\end{itemize}
\end{lemma}

Now, we are ready to present an intrapersonal equilibrium of investor 2. 

\begin{theorem}
\label{thm:follower}
An intrapersonal equilibrium of investor 2 (Definition~\ref{def:intra E investor 2}) is given by 
\begin{align}\label{hat u_2}
\hat u_2(t,p) &:= \frac{2}{2-{\lambda_2}}  \hat u(t,p) + \frac{\lambda_2}{2-\lambda_2}u_1(\delta(t))\\
&= \frac{2}{2-{\lambda_2}} \left[\frac{\theta(p)-r}{\sigma^2\gamma_2} - \frac{\beta(p)\partial_p a_2(t,p)}{\sigma}+\frac{\lambda_2}{2}u_1(\delta(t))\right],\ \ \forall (t,p)\in [0,T]\times (0,1), \label{hat u_2'}
\end{align}
where $\theta,\beta$ are given by \eqref{theta} and $a_2\in C^{1,2}([0,T)\times(0,1))\cap C([0,T]\times(0,1))$ is the unique solution to \eqref{Cauchy 1}. The value function under \(\hat u_2\) is independent of $\alpha\in\operatorname{Step}([0,T];\R)$ and takes the form
\begin{align}\label{equivalue2}
J_2^{(\alpha)}(t,\bm x,p;\hat u_2) = \left(1-\frac{\lambda_2}{2}\right)x_2 - \frac{\lambda_2}{2}x_1 + A_2(t,p),\quad \forall (t,\bm x,p)\in[0,T]\times\R^2\times (0,1), 
\end{align}
where $A_2\in C^{1,2}([0,T)\times(0,1))\cap C([0,T]\times(0,1))$ is the unique solution to \eqref{Cauchy 2}.
\end{theorem}

\begin{proof}
Recall $\hat u$ in \eqref{hat u}. Motivated by the last relation in \eqref{z_2}, we define 
$\hat u_2(t,p) := \frac{2}{2-{\lambda_2}} \big[\hat u(t,p)+\frac{\lambda_2}{2}u_1(\delta(t))\big]$, 
which is simply \eqref{hat u_2}. For any $\alpha\in \operatorname{Step}([0,T];\R)$, $\hat u_2^{(\alpha)}(t,p) = \frac{2}{2-{\lambda_2}} \big[\hat u(t,p)+\frac{\lambda_2}{2}\alpha(t)\big]$, thanks to Remark~\ref{rem:G-adapted} and $\hat u^{(\alpha)}=\hat u$ (Lemma~\ref{lem:hat u} (i)). As $\hat u$ is bounded (Lemma~\ref{lem:hat u} (i)), we have $\int_0^T |\hat u^{(\alpha)}_2(s,P(s))|^2\, ds<\infty$ $\P$-a.s., whence $\hat u_2\in\A_2$ (Definition~\ref{def:admissible strategies}) and it is Markov by construction. For any $(t,\bm x,p)\in [0,T]\times\R^2\times(0,1)$, set $z_2\in\R$ as in \eqref{z_2}. Note that \eqref{J^alpha-J^alpha<=} directly implies \eqref{intra E relation u}, which turns into \eqref{intra E relation} by \eqref{J_2=cJ_2} (with $(u,u_2)$ replaced by $(\hat u,\hat u_2)$). Hence, $\hat u_2$ is an intrapersonal equilibrium of investor 2. Finally, \eqref{equivalue2} follows from \eqref{J_2=cJ_2}, Lemma~\ref{lem:hat u} (iii), \eqref{V_2 and g_2}, and \eqref{z_2}. 
\end{proof}

Recently, \cite{huang2026partial} has identified that under partial information, the portfolio of a mean-variance investor consists of the ``{\it myopic trading}'' and ``{\it intertemporal hedging}'' components associated with the posterior probability process $\mathbb{P}(\mu=\mu_1\mid\mathcal{F}^S_\cdot)$, characterized by $P(\cdot)$ in \eqref{P}. The strategy $\hat u_2$ in Theorem~\ref{thm:follower} extends this interpretation to a Stackelberg game with asymmetric information, through a new ``{\it information learning}'' component. As explained below \cite[Remark 8]{huang2026partial}, given $P(t)=p$, the first term in \eqref{hat u_2'} represents investor 2' {\it myopic demand}, which optimizes trading as if the currently estimated stock return $\theta(p)$ were the true return $\mu$; the second term in \eqref{hat u_2'} captures investor 2' {\it hedging demand}, which hedges against future fluctuations in \(P(\cdot)\), and thus in $\theta(P(\cdot))$, the estimated stock return. The third term in \eqref{hat u_2'}, new to the literature, reflects investor 2's {\it learning demand}, which adds the observed trading $u_1(\delta(t))$ of investor 1 (who knows the true stock return $\mu$) to his portfolio, adjusted suitably by his relative performance coefficient $\lambda_2$. 


\begin{remark}\label{rem:indep. of alpha}
While $\hat u_2$ in \eqref{hat u_2} actively tracks investor 1's trading actions $u_1(\delta(\cdot))$, the induced value function \eqref{equivalue2}, somewhat surprisingly, is independent of $u_1(\delta(\cdot))$ (or its realized path $\alpha$). Recall that we first found an intrapersonal equilibrium for $\cJ_2^{(\alpha)}(t,z_2,p; u)$ in \eqref{cJ_2^alpha} that is independent of $u_1(\delta(\cdot))$---namely, $\hat u$ in \eqref{hat u} (see Lemma~\ref{lem:hat u} (iii)-(iv)); next, we obtained $\hat u_2$ for $J_2^{(\alpha)}(t,\bm x,p; u)$ in \eqref{J_2^alpha}, the original problem, via the linear relation
$\hat u(t,p) = \big(1-\frac{\lambda_2}{2}\big)\hat u_2(t,p) - \frac{\lambda_2}{2}u_1(\delta(t))$ 
in \eqref{z_2}, which connects \eqref{cJ_2^alpha} and \eqref{J_2^alpha}. 
That is, $\hat u_2$ actively tracks $u_1(\delta(\cdot))$ only to offset it entirely, ensuring that the right-hand side of the linear relation as a whole is independent of \(u_1(\delta(t))\), so as to recover $\hat u$ (which then yields the equilibrium property and an $\alpha$-independent value function). 
\end{remark}

\begin{remark}
When \(\lambda_2=0\), \eqref{hat u_2} reduces to $\hat u_2(t,p) = \frac{\theta(p)-r}{\sigma^2\gamma_2} - \frac{\beta(p)\partial_p a_2(t,p)}{\sigma}$, which is an intrapersonal equilibrium for a single mean-variance investor under partial information; see \cite[equation (51)]{huang2026partial}. Indeed, as $\lambda_2 =0$ means that investor 2 has no relative performance concern, investor 1's wealth no longer enters his objective \eqref{J_2^alpha}. Investor 2 thus behaves as a single agent.
\end{remark}


\section{The Leader's Optimization Problem}\label{sec:leader}
In this section, we focus on finding investor 1's intrapersonal equilibrium of order $K$ (Definition~\ref{def:intra E investor 1}) w.r.t.\ $\hat u_2$ in \eqref{hat u_2}, an intrapersonal equilibrium of investor 2 (Theorem~\ref{thm:follower}). First, we take up the idealized case of {\it continuous} action sampling, in the form of {\it exploratory} stochastic control  recently developed in continuous-time RL; see, e.g., \cite{wang2020reinforcement,wang2020continuous,dai2023learning}. While one cannot realistically perform continuous sampling (as discussed below Remark~\ref{rem:no distinction}), this idealized case allows an explicit intrapersonal analysis; see Theorem~\ref{thm:leader}. Next, through delicate approximations (Lemma~\ref{lem:(J-J)-(J-J)}), investor 1 can approach the limiting idealized case by sampling actions on a time grid $\D$ dense enough, thereby obtaining an intrapersonal equilibrium of order $K$; see Theorem~\ref{thm:final}. 

For any $\Pi\in\A_1$ and time grid $\D$ of $[0,T]$, we recall \eqref{dynamics investor 1} and define, in line with \eqref{Z_2 def}, 
\begin{align}\label{transz}
	Z_1(s) := X_1(s) - \lambda_1 \overline{X}(s) = \left(1-\frac{\lambda_1}{2}\right)X_1(s) - \frac{\lambda_1}{2}X_2(s). 
\end{align}
Thanks to \eqref{dynamics investor 1} and \eqref{hat u_2}, $Z_1$ satisfies
\begin{equation}\label{Z_1}
d Z_1(s) = \bigg(\chi u_1(\delta(s)) - \frac{\lambda_1\hat u(s,P(s))}{2-\lambda_2} \bigg) \Big((\mu-r)\, ds+ \sigma\, dW(s)\Big),\quad Z_1(t)=z_1,
\end{equation}
where \(P(\cdot)\) is given as in \eqref{dynamics investor 1} and we set
\begin{equation}\label{z_1}
z_1 = \bigg(1-\frac{\lambda_1}{2}\bigg)x_1 - \frac{\lambda_1}{2}x_2\qquad \hbox{and}\qquad \chi := \bigg(1-\frac{\lambda_1}{2}\bigg)-\frac{\lambda_1}{2}\frac{\lambda_2}{2-\lambda_2}= \frac{2-\lambda_2-\lambda_1}{2-\lambda_2}.
\end{equation}
Conversely, for any $(t,z_1, p)\in [0,T]\times\R\times(0,1)$ and $\Pi\in\A_1$, we consider $Z_1$ in \eqref{Z_1} and define
\begin{equation}\label{cJ_1^D}
	\cJ_1^\D(t,z_1,p;\Pi) := \widetilde{\E}_{t,z_1,p}\bigg[ Z_1(T) + \lambda_0\sum_{i:\,t_i\in[t,T)}H(\Pi_{t_i})(t_{i+1}-t_i) \bigg] - \frac{\gamma_1}{2}\widetilde{\operatorname{Var}}_{t,z_1,p}\left[Z_1(T)\right],
	\end{equation}
where the subscript $(t,z_1,p)$ denotes conditioning on ``$Z_1(t)=z_1$ and $P(t)=p$.'' Under the first relation in \eqref{z_1}, we deduce from \eqref{J_1^D} and \eqref{cJ_1^D} that 
\begin{equation}\label{J_1^D=cJ_1^D}
J_1^\D(t,\bm x,p;\Pi) = \mathcal J_1^\D(t,z_1,p;\Pi).
\end{equation}

\subsection{The Idealized Case of Continuous Sampling}\label{subsec:exploratory}
Given $\Pi\in\A_1$, when investor 1's discrete sampling of actions becomes increasingly frequent (i.e., $|\D|\to 0$), we find that, by following the same derivation in \cite[Section 2.2]{dai2023learning}, $X_1$ in \eqref{dynamics investor 1} turns into
\begin{align}\label{fX_1}
	d \fX_1(s) = \mathfrak m^{\Pi_s} \big((\mu-r)\, ds + \sigma dW(s)\big) +\sigma \mathfrak s^{\Pi_s}\,d\overline W(s),\qquad \fX_1(t)=x_1,
\end{align} 
where $\fm^\nu, \fs^{\nu}$ are defined in \eqref{tb} and \(\overline W\) is a Brownian motion independent of \(W\) (which accounts for additional randomness from action sampling, independent of the market-driving $W$; see  \cite[Section 2.2]{dai2023learning}). 
Similarly, under $\hat u_2$ in \eqref{hat u_2}, we find that as $|\D|\to 0$, $X_2$ in \eqref{dynamics investor 1} turns into
\begin{align}\label{fX_2}
	d\fX_2(s)	&= \bigg(\frac{2\hat u(s,P(s))}{2-\lambda_2}+\frac{\lambda_2\fm^{\Pi_s}}{2-\lambda_2}\bigg) \Big((\mu-r)\, ds + \sigma dW(s)\Big) + \sigma\frac{\lambda_2\,\fs^{\Pi_s}}{2-\lambda_2}\,d\overline W(s),\quad \fX_2(t)=x_2.
\end{align}
We then correspondingly modify the $\D$-dependent mean-variance objective \eqref{J_1^D} to 
\begin{align}\label{J_1}
J_1(t,\boldsymbol x,p;\Pi) &:=\mathbb E_{t,\boldsymbol{x},p}\bigg[\fX_1(T) -\lambda_1\overline{\fX}(T) +\lambda_0\int_t^T H(\Pi_s)\,ds\bigg] \notag\\
&\qquad-\frac{\gamma_1}{2}\operatorname{Var}_{t,\boldsymbol{x},p}\big[\fX_1(T)-\lambda_1\overline{\fX}(T)\big],
\end{align}
with $\overline{\fX}(T):= \big(\fX_1(T)+\fX_2(T)\big)/2$.

\begin{remark}\label{rem:connection to RL}
Note that \eqref{J_1} (under $(\fX_1,\fX_2,P)$ in \eqref{fX_1}, \eqref{fX_2}, and \eqref{dynamics investor 1}) aligns with exploratory stochastic control for continuous-time RL (introduced in \cite{wang2020reinforcement} and has gained significant traction; see e.g., \cite{wang2020continuous, Guo22, Tang22, dai2023learning, JZ22-PG, JZ23, HWZ25, HYZ26}). In particular, \cite{wang2020continuous} derives an exploratory pre-committed strategy for a single mean-variance investor; \cite{dai2023learning} obtains an exploratory intrapersonal equilibrium for a single investor who applies the mean-variance criterion to log returns. By contrast, \eqref{J_1} applies the mean-variance criterion to investor 1's mixed wealth, which depends on investor 2's wealth. 
\end{remark}

\begin{remark}\label{rem:different interpretations}
Despite the connection in Remark~\ref{rem:connection to RL}, the stochastic control framework for RL and \eqref{J_1} have distinct interpretations. In the former, a stochastic policy $\Pi$ represents exploration in RL and the entropy term $\int_t^T H(\Pi_s)\,ds$ measures the degree of exploration; in the latter, investor 1 randomizes her actions via $\Pi$ to hide her private information from investor 2, and the entropy term quantifies how intensely private information is guarded (through randomization).
\end{remark}

We next define an exploratory intrapersonal equilibrium analogously to \cite[Definition 2.2]{dai2023learning}.

\begin{definition}
\label{def2.5}
Given investor 2's intrapersonal equilibrium \(\hat u_2\) in \eqref{hat u_2}, we say \(\hat\Pi\in\mathcal A_1\) is an exploratory intrapersonal equilibrium of investor 1 (w.r.t.\ $\hat u_2$), if it is Markov and for every \((t,\bm{x},p)\in[0,T)\times\R^2\times (0,1)\) and $\nu\in \cP_2(\R)$, 
\begin{equation}\label{intra E investor 1 idealized}
\limsup_{h\downarrow0} \frac{ J_1\big(t,\bm x,p;\hat \Pi^{h,\nu}\big) - J_1\big(t,\bm x,p;\hat \Pi\big)}{h} \le 0.
\end{equation}
\end{definition}

To find an exploratory intrapersonal equilibrium of investor 1, 
we introduce, similarly to \eqref{Z_2 def},   
\begin{align}\label{transz}
	\mathcal Z_1(s) :=  \fX_1(s) - \lambda_1 \overline{\fX}(s) = \left(1-\frac{\lambda_1}{2}\right)\fX_1(s) - \frac{\lambda_1}{2}\fX_2(s). 
\end{align}
Thanks to \eqref{fX_1} and \eqref{fX_2}, $\cZ_1$ satisfies
\begin{align}\label{cZ_1}
d \cZ_1(s) = \left(\chi \fm^{\Pi_s} - \frac{\lambda_1 \hat u(s,P(s))}{2-\lambda_2} \right) \Big(  (\mu-r)\, ds+\sigma\, dW(s)\Big) + \sigma\chi\,\fs^{\Pi_s}\, d\overline W(s),\quad \cZ_1(t)=z_1,
\end{align}
where \(P(\cdot)\) is given as in \eqref{dynamics investor 1} and we take $z_1$ and $\chi$ as in \eqref{z_1}. 
Conversely, for any $(t,z_1, p)\in [0,T]\times\R\times(0,1)$ and $\Pi\in\A_1$, we consider $\cZ_1$ in \eqref{cZ_1} and define
\begin{equation}\label{cJ_1}
	\cJ_1(t,z_1,p;\Pi) := \mathbb E_{t,z_1,p}\left[ \cZ_1(T) + \lambda_0\int_t^T H(\Pi_s)\,ds \right] - \frac{\gamma_1}{2}\operatorname{Var}_{t,z_1,p}\left[\cZ_1(T)\right],
	\end{equation}
where the subscript $(t,z_1,p)$ denotes conditioning on ``$\cZ_1(t)=z_1$ and $P(t)=p$.'' Under the first relation in \eqref{z_1}, we deduce from \eqref{J_1} and \eqref{cJ_1} that 
\begin{equation}\label{J_1=cJ_1}
J_1(t,\bm x,p;\Pi) = \mathcal J_1(t,z_1,p;\Pi).
\end{equation}
Hence, the relation \eqref{intra E investor 1 idealized} can be restated as 
\begin{align}\label{intra E investor 1 idealized cJ} 
\limsup_{h\downarrow0} \frac{ \cJ_1\big(t,z_1,p;\hat \Pi^{h,\nu}\big) - \cJ_1\big(t,z_1,p;\hat\Pi\big)}{h} \le 0. 
\end{align}
To find such a $\hat \Pi$, we will again follow \cite{bjork2017time} to derive the associated extended HJB equation. 

Consider a candidate $\hat \Pi$ and its value function $\cV_1(t,z_1,p) := \cJ_1(t,z_1,p;\hat\Pi)$. Following the derivation in \cite{bjork2017time} under the dynamics of \((\cZ_1,P)\) in \eqref{cZ_1} and \eqref{dynamics investor 1}, we obtain
\begin{align}\label{3-1}
0 = \partial_t \cV_1 + \sup_{\nu\in\cP_2(\R)}&\bigg\{(\mu-r)\left(\chi\fm^\nu - \frac{\lambda_1\hat u}{2-\lambda_2}\right)\partial_{z_1}\cV_1+\frac{\sigma^2}{2}\left[\left(\chi\fm^\nu-\frac{\lambda_1\hat u}{2-\lambda_2}\right)^2+\chi^2(\fs^\nu)^2
\right]\partial_{z_1z_1}\cV_1 \notag\\
&+\left[\frac{\mu-\mu_2}{\sigma}-\frac{\mu_1-\mu_2}{\sigma}p\right]\beta(p)\partial_p\cV_1+\frac{\beta^2(p)}{2}\partial_{pp}\cV_1+\sigma\left(\chi\fm^\nu-\frac{\lambda_1\hat u}{2-\lambda_2}\right)\beta(p)\partial_{z_1p}\cV_1 \notag\\
&-\frac{\gamma_1}{2}\sigma^2\left[\left(\chi\fm^\nu-\frac{\lambda_1\hat u}{2-\lambda_2}\right)^2+\chi^2(\fs^{\nu})^2\right](\partial_{z_1}g_1)^2 -\frac{\gamma_1}{2}\beta^2(p)(\partial_p g_1)^2 \notag\\
&-\gamma_1\sigma\left(\chi\fm^\nu-\frac{\lambda_1\hat u}{2-\lambda_2}\right)\beta(p)\partial_p g_1\partial_{z_1}g_1+\lambda_0H(\nu)\bigg\}(t,z_1,p),
\end{align}
with terminal condition \(\cV_1(T,z_1,p)=z_1\), where $g_1(s,z_1,p)$ is the solution to
\begin{align}\label{3-2}
0=\partial_t &g_1+(\mu-r)\left(\chi\fm^{\hat\Pi_t}-\frac{\lambda_1\hat u}{2-\lambda_2}\right)\partial_{z_1}g_1+\frac12\sigma^2\left[\left(\chi\fm^{\hat\Pi_t}-\frac{\lambda_1\hat u}{2-\lambda_2}\right)^2+\chi^2(\fs^{\hat \Pi_t})^2\right]\partial_{z_1z_1}g_1 \notag\\
&+\left[\frac{\mu-\mu_2}{\sigma}-\frac{\mu_1-\mu_2}{\sigma}p\right]\beta(p)\partial_p g_1 + \frac12\beta^2(p)\partial_{pp}g_1 +\sigma\left(\chi\fm^{\hat\Pi_t}-\frac{\lambda_1\hat u}{2-\lambda_2}\right)
\beta(p)\partial_{z_1p}g_1,
\end{align}
with terminal condition \(g_1(T,z_1,p)=z_1\). Our goal is to find a solution $(\cV_1, g_1)$ to \eqref{3-1}-\eqref{3-2} such that the maximizer of the supremum in \eqref{3-1} coincides with $\hat \Pi_t$. 

The next result (whose proof is relegated to Appendix~\ref{subsec:proof of lem:hat Pi}) obtains such a solution, which is, interestingly, independent of $p$. It further establishes concrete stochastic representations, which ultimately lead to a zeroth-order estimate for $J_1(t,\bm x,p;\Pi)$ (i.e., \eqref{J_1(h)-J_1}). 

\begin{lemma}\label{lem:hat Pi}
For any $(t,z_1,p)\in [0,T]\times\R\times (0,1)$, consider
\begin{align}
\cV_1(t,z_1) &:= z_1 + \left[ \frac{(\mu-r)^2}{2\sigma^2\gamma_1} + \frac{\lambda_0}{2} \log\left(\frac{2\pi \lambda_0}{\gamma_1\sigma^2\chi^2}\right)\right](T-t),\label{cV_1}\\
 g_1(t,z_1) &:= z_1 + \frac{(\mu-r)^2}{\sigma^2\gamma_1}(T-t),\label{g_1}\\
 \hat\Pi_t \sim \mathcal{N}\bigg(\frac{\mu-r}{\sigma^2 \gamma_1 \chi}&+\frac{\lambda_1\hat u(t,p)}{2-\lambda_2-\lambda_1},\,\frac{\lambda_0}{\gamma_1\sigma^2\chi^2}\bigg)\label{hat Pi}\\
= \mathcal{N}\bigg(\frac{\mu-r}{\sigma^2 \gamma_1 \chi}&+\frac{\lambda_1}{2-\lambda_2-\lambda_1}\bigg[\frac{\theta(p)-r}{\sigma^2\gamma_2}-\frac{\beta(p)}{\sigma}\partial_p a_2(t,p) \bigg],\,\frac{\lambda_0}{\gamma_1\sigma^2\chi^2}\bigg),\label{hat Pi'}
\end{align}
where $\mathcal N(a,b)$ is the Gaussian distribution with mean $a$ and variance $b$, the maps $\theta,\beta$ are given by \eqref{theta} and $a_2\in C^{1,2}([0,T)\times(0,1))\cap C([0,T]\times(0,1))$ is the unique solution to \eqref{Cauchy 1}. Then,  
\begin{itemize}
\item [(i)] $(\cV_1, g_1)$ and $\hat \Pi$ satisfy \eqref{3-1}-\eqref{3-2} and the maximizer of the supremum in \eqref{3-1} coincides with $\hat \Pi_t$. 
\item [(ii)] Let $\hat \cZ_1$ denote \eqref{cZ_1} with $\Pi=\hat\Pi$. Then, $g_1(t,z_1)=\mathbb E_{t,z_1}[\hat \cZ_1(T)]$ and 
\[
 \cV_1(t,z_1) = \E_{t,z_1}\left[\hat \cZ_1(T)+\lambda_0\int_t^T H(\hat \Pi_s)\,ds\right]- \frac{\gamma_1}{2}\operatorname{Var}_{t,z_1}[\hat \cZ_1(T)] = \cJ_1(t,z_1,p;\hat\Pi),\ \ \forall p\in(0,1). 
\]
\item [(iii)] For any $(t,\bm x,p)\in [0,T]\times\R^2\times (0,1)$, $\nu\in\cP_2(\R)$, and $h>0$, $\hat\Pi$ satisfies
\begin{equation}\label{J_1(h)-J_1}
J_1(t,\bm x,p;\hat \Pi^{h,\nu}) - J_1(t,\bm x,p;\hat \Pi) \le  C_0\big(1+\fm^\nu+\fs^\nu\big)^2 \big(h^{3/2} + {h^2}\big),
\end{equation}
where $C_0>0$ is independent of $(t,\bm x,p)$, $\nu$, and $h$. 
\end{itemize}
\end{lemma}

Now, we are ready to present an exploratory intrapersonal equilibrium of investor 1. 

\begin{theorem}\label{thm:leader}
$\hat\Pi$ in \eqref{hat Pi} is an exploratory intrapersonal equilibrium of investor 1 (w.r.t.\ $\hat u_2$ in \eqref{hat u_2}); recall Definition~\ref{def2.5}. The value function under $\hat\Pi$ is independent of $p$ and takes the form
\begin{align}\label{eq:leader-value}
J_1(t,x_1,x_2,p;\hat\Pi) = \left(1-\frac{\lambda_1}{2}\right)x_1 - \frac{\lambda_1}{2}x_2 + \left[\frac{(\mu-r)^2}{2\sigma^2\gamma_1} + \frac{\lambda_0}{2}\log\left(\frac{2\pi \lambda_0}{\gamma_1\sigma^2\chi^2}\right)\right](T-t),
\end{align}
for all $(t, x_1,x_2)\in[0,T]\times\R^2$, with $\chi$ defined as in \eqref{z_1}. 
\end{theorem}

\begin{proof}
As \eqref{J_1(h)-J_1} directly implies \eqref{intra E investor 1 idealized}, $\hat\Pi$ is an exploratory intrapersonal equilibrium of investor 1. Also, \eqref{eq:leader-value} follows from \eqref{J_1=cJ_1}, Lemma~\ref{lem:hat Pi} (ii), and the first relation in \eqref{z_1}. 
\end{proof}

The exploratory intrapersonal equilibrium $\hat\Pi$, given by \eqref{hat Pi}, has several interesting features. First, its Gaussian mean has two separate components. The first one, $\frac{\mu-r}{\sigma^2\gamma_1\chi}$, reflects investor 1's information advantage: as she knows $\mu$, she takes the standard mean-variance trading $\frac{\mu-r}{\sigma^2\gamma_1}$ (which is an intrapersonal equilibrium for a single mean-variance investor; see \cite[Remark 4]{huang2026partial}), adjusted by $\chi$ in \eqref{z_1} depending on the relative performance coefficients $\lambda_1, \lambda_2$. The second component, $\frac{\lambda_1\hat u(t,p)}{2-\lambda_2-\lambda_1}$, reflects investor 1's anticipation of investor 2's strategy. 
Indeed, $\hat u$ here comes from investor 2's intrapersonal equilibrium $\hat u_2$ in \eqref{hat u_2}. That is, the Gaussian mean in $\hat\Pi$ combines the leader's (investor 1's) information advantage and strategic internalization of the follower's (investor 2's) reaction, both suitably adapted by the relative performance coefficients $\lambda_1,\lambda_2$. 

On the other hand, the Gaussian variance in $\hat\Pi$, $\frac{\lambda_0}{\gamma_1\sigma^2\chi^2}$, directly depends on the weight $\lambda_0$ placed on the entropy term in \eqref{J_1}. The larger the weight, the more intensely investor 1 hopes to randomize her actions (to better protect her private information of $\mu$, which is part of the Gaussian mean in $\hat\Pi$), whence a larger Gaussian variance that encourages dispersion of actions. 


Finally, the value function \eqref{eq:leader-value} under $\hat\Pi$ is independent of \(p=P(t) = \P(\mu=\mu_1\mid \mathcal F^S_t)\), although the Gaussian mean in $\hat\Pi$ depends on $p$. As just mentioned, the Gaussian mean depends on \(p\) because investor 1 internalizes investor 2's strategy $\hat u_2$ in \eqref{hat u_2}, which depends on $p$ as investor 2 is unaware of the true value of $\mu$. In fact, the $p$-dependence of the Gaussian mean in $\hat\Pi$ is chosen specifically to offset the \(p\)-dependence of $\hat u_2$, which yields the intrapersonal equilibrium property (see the proof of Lemma \ref{lem:hat Pi} in Appendix~\ref{subsec:proof of lem:hat Pi}) and a value function independent of $p$. This is analogous to investor 2's optimization problem---his strategy $\hat u_2(t,p)$ depends on investor 1's trading action $u_1(\delta(t))$ but his value function is independent of the path of $u_1(\delta(\cdot))$; recall Remark~\ref{rem:indep. of alpha}. 


\begin{remark}
If $\lambda_1=0$ (i.e., investor 1 has no relative performance concern at all), then \(\chi=1\) in \eqref{z_1} and $\hat\Pi$ in \eqref{hat Pi} reduces to $\hat\Pi_t \sim \mathcal{N}\big(\frac{\mu-r}{\sigma^2 \gamma_1},\,\frac{\lambda_0}{\gamma_1\sigma^2}\big)$, 
which is an exploratory intrapersonal equilibrium for a single investor who applies the mean-variance criterion to her wealth. 
This result is new to the literature of continuous-time RL, as prior studies either focus on pre-committed strategies or apply the mean-variance criterion to an investor's log returns; see Remark~\ref{rem:connection to RL}. 
\end{remark}

\begin{remark}
While the exploratory framework allows tractable analysis (which yields Theorem~\ref{thm:leader}), our Stackelberg game cannot be stated under the exploratory wealth $\fX_1$ in \eqref{fX_1}. 
Note that \eqref{fX_1} does not involve actual trading actions: it simply captures the limiting effect of ever more frequent action sampling. Hence, for investor 2 to observe and react to investor 1's trading actions, we need the $\D$-dependent wealth process $X_1$ in \eqref{eq2-14}, where actual trading actions $u_1(\delta(s))$ are present. 
\end{remark}


\subsection{Finding an Intrapersonal Equilibrium of Order $K$}\label{subsec:order K}
In view of the zeroth-order estimate \eqref{J_1(h)-J_1}, it is natural to ask whether $\hat\Pi$ in \eqref{hat Pi}, obtained under the idealized exploratory framework in Section~\ref{subsec:exploratory}, actually satisfies \eqref{intra E investor 1} for some $K>0$, so that it is in fact a realistic ``intrapersonal equilibrium of order $K$'' for investor 1. 
To answer this, we start with comparing 
the left-hand sides of \eqref{intra E investor 1} and \eqref{J_1(h)-J_1} in the next result. 

\begin{lemma}\label{lem:(J-J)-(J-J)}
There exist constants $C_1,C_2>0$, independent of any time grid $\D=\{t_i\}_{i=0}^n$ of $[0,T]$, $(t_i,\bm x,p)\in \D\setminus\{T\}\times \R^2 \times (0,1)$, and $\nu\in\cP_2(\R)$, such that
\begin{align}\label{(J-J)-(J-J) to prove}
&\left|   \left(J_1^\D\big(t_i,\bm x,p;\hat \Pi^{h_i,\nu}\big) - J_1^\D\big(t_i,\bm x,p;\hat \Pi\big)\right)  - \left(J_1\big(t_i,\bm x,p;\hat \Pi^{h_i,\nu}\big) - J_1\big(t_i,\bm x,p;\hat \Pi\big) \right)   \right|\notag\\
&\hspace{1in}\le C_1\big(1+\big|\fm^\nu-\fm^{\hat\Pi_{t_i}}\big|\big) |\D|^{3/2} + C_2 (1+(\fs^\nu)^2) |\D|^2,\quad \hbox{with}\ \ h_i:= t_{i+1}-t_i. 
\end{align}
\end{lemma}

The proof of Lemma~\ref{lem:(J-J)-(J-J)} (relegated to Appendix~\ref{subsec:proof of lem:(J-J)-(J-J)}) carefully decomposes $J_1^\D\big(t_i,\bm x,p;\hat \Pi^{h_i,\nu}\big) - J_1^\D\big(t_i,\bm x,p;\hat \Pi\big)$ and $J_1\big(t_i,\bm x,p;\hat \Pi^{h_i,\nu}\big) - J_1\big(t_i,\bm x,p;\hat \Pi\big)$ into a number of terms and estimates the term-by-term differences, relying specifically on the mean-variance structure and Malliavin calculus.  

\begin{remark}\label{rem:Jia}
The approximation result \cite[Theorem 4.2]{jia2025accuracy} cannot be directly applied here, as the coefficient $\beta(p)= \frac{\mu_1-\mu_2}{\sigma} p(1-p)$ of $P(\cdot)$ in \eqref{dynamics investor 1} is degenerate and the mean-variance objective \eqref{J_1^D} has unbounded payoffs. In fact, even if \cite[Theorem 4.2]{jia2025accuracy} were applicable, it would only give $|\D|^{1/2}$ on the right-hand side of \eqref{(J-J)-(J-J) to prove} (weaker than $|\D|^{3/2}$ we have). With an exponent less than 1 (in this case $1/2$), we would not be able to achieve \eqref{intra E investor 1} in Theorem~\ref{thm:final} below. 
\end{remark}

From Lemma~\ref{lem:(J-J)-(J-J)} and \eqref{J_1(h)-J_1} comes the main result of this paper: 
$\hat\Pi$ in \eqref{hat Pi} is investor 1's intrapersonal equilibrium of order $K$, hence $(\hat\Pi, \hat u_2)$ is a Stackelberg equilibrium for the two investors.  

\begin{theorem}\label{thm:final}
There exist constants $\widetilde C_1,\widetilde C_2>0$, independent of the time grid $\D=\{t_i\}_{i=0}^n$ of $[0,T]$, $(t_i,\bm x,p)\in \D\times \R^2 \times (0,1)$, and $\nu\in\cP_2(\R)$, such that $\hat\Pi$ in \eqref{hat Pi} satisfies
\begin{align}\label{J-J}
J_1^\D\big(t_i,\bm x,p;\hat \Pi^{h_i,\nu}\big) - J_1^\D\big(t_i,\bm x,p;\hat \Pi\big) &\le \widetilde C_1\left(\big(1+ \fm^\nu+\fs^\nu\big)^2+|\fm^\nu-\fm^{\hat\Pi_{t_i}}|\right) |\D|^{3/2}\notag\\
&\hspace{1.25in}+\widetilde C_2\big(1+\fm^\nu+\fs^\nu\big)^2  {|\D|^2}, 
\end{align}
with $h_i:= t_{t+1}-t_i$, 
for any time grid $\D$, $(t_i,\bm x,p)\in \D\setminus\{T\}\times \R^2 \times (0,1)$ and $\nu\in\cP_2(\R)$. Hence, for any $K>0$, \eqref{intra E investor 1} is fulfilled with 
\begin{equation}\label{C,kappa}
C = (\widetilde C_1+\widetilde C_2\sqrt{T}) \bigg(1+\frac{|\mu-r|}{\sigma^2\gamma_1\chi}+\frac{\lambda_1\|\hat u\|_{\infty}}{2-\lambda_2-\lambda_1} + \frac{1}{\sigma\chi}\bigg(\frac{\lambda_0}{\gamma_1}\bigg)^{1/2}+2K\bigg)^2 +  \widetilde C_1 K \quad\hbox{and}\quad \kappa=\frac12. 
\end{equation}
That is, for any $K>0$, $\hat\Pi$ is investor 1's intrapersonal equilibrium of order $K$ (w.r.t.\ $\hat u_2$ in \eqref{hat u_2}) and $(\hat\Pi, \hat u_2)$ forms a Stackelberg equilibrium between the two investors (recall Definition~\ref{def:intra E investor 1}).
\end{theorem}

\begin{proof}
By applying the estimate \eqref{J_1(h)-J_1} to Lemma~\ref{lem:(J-J)-(J-J)}, we directly obtain \eqref{J-J}. Note that the right-hand side of \eqref{J-J} is bounded from above by
$L 
 := \big[ (\widetilde C_1+\widetilde C_2\sqrt{T})\big(1+ \fm^\nu+\fs^\nu\big)^2+ \widetilde C_1 |\fm^\nu-\fm^{\hat\Pi_{t_i}}| \big] |\D|^{3/2}.$
For any $K>0$, ``$\nu\in B_K(\hat\Pi_{t_i})$'' implies, thanks to \eqref{W_2}, ``$|\fm^\nu - \fm^{\hat\Pi_{t_i}}|<K$ and $|\fs^\nu - \fs^{\hat\Pi_{t_i}}|<K$.'' This, along with \eqref{hat Pi}, indicates that $L\le C|\D|^{3/2}$ with $C$ given as in \eqref{C,kappa}. 
We thus obtain \eqref{intra E investor 1} with $(C,\kappa)$ given as in \eqref{C,kappa}. Hence, $\hat\Pi$ is investor 1's intrapersonal equilibrium of order $K$ and thus $(\hat\Pi, \hat u_2)$ forms a Stackelberg equilibrium between the two investors.  
\end{proof}

\appendix
    
\section{Proofs}
\subsection{Proof of Lemma~\ref{lem:a_2 bdd}}\label{subsec:proof of lem:a_2 bdd}
(i) By \cite[Lemma 3 (i)]{huang2026partial} (with $\eta(\cdot)$ and $\kappa_i$ therein taken to be $0$ and $1/\gamma_2$, respectively), we immediately have the existence of a unique solution  $a_2\in C^{1,2}([0,T)\times(0,1))\cap C([0,T]\times(0,1))$ to \eqref{Cauchy 1}; moreover, $a_2$ is bounded on $[0,T]\times (0,1)$ and satisfies 
\begin{equation}\label{a_2=E}
a_2(t,p) = \frac{1}{\gamma_2}\E_{t,p}^\Q\bigg[\int_t^T \bigg(\frac{\theta(P(s))-r}{\sigma}\bigg)^2\, ds\bigg],
\end{equation}
where $\E^\Q$ is taken under $\Q\approx\P$ with $\frac{d\Q}{d\P}:= \exp\big(-\frac12\int_t^T\big(\frac{\theta(P(s)-r)}{\sigma}\big)^2\, ds-\int_t^T\frac{\theta(P(s)-r)}{\sigma}\, d\widehat W(s)\big),
$
where $P$ fulfills the third line of \eqref{dynamics investor 2}. Note that $\widehat W_\Q(s) := \widehat W(s) + \int_t^s \frac{\theta(P(\bar s))-r}{\sigma}\, d\bar s$, $t\le\bar s\le T$, is a Brownian motion under $\Q$, and the dynamics of $P$ under $\Q$ becomes
\begin{equation}\label{P under Q}
dP(s) = -\beta(P(s))\frac{\theta(P(s))-r}{\sigma}\, ds + \beta(P(s))\, d\widehat W_\Q(s),\quad P(t)=p.
\end{equation}

By \cite[Theorem 10, p.\ 72]{Friedman-book-64}, $a_2$ is infinitely differentiable w.r.t.\ $p$. Let us proceed to show that all these partial derivatives w.r.t.\ $p$ is bounded. 
We will write $P^{t,p}$ for $P$ in \eqref{P under Q} to stress ``$P(t)=p$.'' By \cite[Theorem 5.3, p.\ 120]{friedman1975stochastic}, $(P^{t,p+h}(s)-P^{t,p}(s))/h\to \zeta(s)$ in $L^2(\Omega,\mathcal F^S_T,\Q)$ as $h\to 0$ for all $s\in[t,T]$, where $\zeta$ is the unique strong solution to 
\begin{equation}\label{zeta}
d\zeta(s) = \zeta(s)\Gamma(P^{t,p}(s))\, ds + \zeta(s)\Lambda(P^{t,p}(s))\, d\widehat W_\Q(s),\quad \zeta(t)=1,
\end{equation}
with $\Gamma(p) := -\frac{d}{dp}(\beta(p)\frac{\theta(p)-r}{\sigma})$ and $\Lambda(p) := \frac{d}{dp}\beta(p)$. As $P(s)\in (0,1)$ $\Q$-a.s.\ and $\Gamma, \Lambda$ are bounded on $(0,1)$, there exists $C_1>0$, depending only on $\|\Gamma\|_\infty$, $\|\Lambda\|_\infty$, and $T$, such that 
\begin{equation}\label{E[sup]}
\E^\Q\bigg[\sup_{t\le s\le T}\zeta(s)^2\bigg] \le C_1 (1+\zeta(t)^2) = 2C_1.
\end{equation}
As shown  in \cite[Appendix A.3 (ii)]{huang2026partial}, one may formally differentiate the right-hand side of \eqref{a_2=E} w.r.t.\ $p$ (i.e., treating $\zeta(s)$ as the derivative of $P^{t,p}(s)$ w.r.t.\ $p$) to obtain 
\begin{equation}\label{a_2'=E}
\partial_p a_2(t,p) = \frac{2(\mu_1-\mu_2)}{\sigma^2\gamma_2}\E_{t,p}^\Q\bigg[\int_t^T \zeta(s)\big(\theta(P(s))-r\big)\, ds\bigg],
\end{equation}
which shows that $\partial_p a_2$ is bounded on $[0,T]\times (0,1)$, by \eqref{E[sup]} and $\theta$ being bounded. Let us write $\zeta^{t,p}$ for $\zeta$ in \eqref{zeta} to stress its dependence on $P^{t,p}$. By \cite[Theorem 5.4, p.\ 122]{friedman1975stochastic}, $(\zeta^{t,p+h}(s)-\zeta^{t,p}(s))/h\to \xi(s)$ in $L^2(\Omega,\mathcal F^S_T,\Q)$ as $h\to 0$ for all $s\in[t,T]$, where $\xi$ is the unique strong solution to 
\begin{equation}\label{xi}
d\xi(s) = \xi(s)\Gamma'(P^{t,p}(s))\, ds + \xi(s)\Lambda'(P^{t,p}(s))\, d\widehat W_\Q(s),\quad \xi(t)=1,
\end{equation}
which also fulfills the estimate \eqref{E[sup]}, with $C_1>0$ replaced by some $C_2>0$ depending on only $\|\Gamma'\|_\infty$, $\|\Lambda'\|_\infty$, and $T$. Hence, by following the same argument in \cite[Appendix A.3 (ii)]{huang2026partial}, we can further differentiate the right-hand side of \eqref{a_2'=E} to obtain
\begin{equation}\label{a_2''=E}
\partial_{pp} a_2(t,p) = \frac{2(\mu_1-\mu_2)}{\sigma^2\gamma_2}\E_{t,p}^\Q\bigg[\int_t^T \xi(s)\big(\theta(P(s))-r\big)+(\mu_1-\mu_2)\zeta^2(s)\, ds\bigg],
\end{equation}
which shows that $\partial_{pp} a_2$ is bounded on $[0,T]\times (0,1)$, as $\zeta$ and $\xi$ fulfill \eqref{E[sup]} (with $C_1$ and $C_2$, respectively) and $\theta$ is bounded. Following the same procedure recursively, we can show that the partial derivative of $a_2$ w.r.t.\ $p$ of any order is bounded.
With $\partial_{p} a_2$ and $\partial_{pp}a_2$ bounded, \eqref{Cauchy 1} directly implies that $\partial_t a_2$ is bounded. With $\partial_{pp} a_2$ and $\partial_{ppp}a_2$ bounded continuous, \eqref{Cauchy 1} again implies that $\partial_p\partial_{t} a_2$ exists and is bounded  {continuous}. As \cite[Theorem 10, p.\ 72]{Friedman-book-64} asserts that $\partial_t\partial_p a_2$ exists and is continuous, Clairaut's theorem implies $\partial_p\partial_{t} a_2 = \partial_t\partial_{p} a_2$ (hence they are bounded). 

(ii) By \cite[Corollary 1]{huang2026partial}, \eqref{Cauchy 2} has a unique solution \(A_2\in C^{1,2}([0,T)\times(0,1))\cap C([0,T]\times(0,1))\), with  $A_2$ and \(\partial_p A_2\) bounded on \([0,T]\times(0,1)\).


\subsection{Proof of Lemma~\ref{lem:hat u}}\label{subsec:proof of lem:hat u}
(i) By \eqref{theta} and the boundedness of \(\partial_pa_2\) (Lemma~\ref{lem:a_2 bdd}), $\hat u$ in \eqref{hat u} is bounded on $[0,T]\times (0,1)$. As $\hat u$ is bounded and independent of $u_1(\delta(\cdot))$, we conclude that $\hat u\in\A_2$ is Markov (Definitions~\ref{def:admissible strategies} and \ref{def:intra E investor 2}) and  $\hat u^{(\alpha)} = \hat u$ for all $\alpha\in\operatorname{Step}([0,T];\R)$ (recall Remark~\ref{rem:G-adapted}).  

(ii) By plugging $g_2$ in \eqref{V_2 and g_2} and $\hat u$ in \eqref{hat u} into \eqref{3-6} and using $\hat u^{(\alpha)} = \hat u$, we obtain 
$0 = \partial_t a_2 + (\theta(p)-r)\hat u + \frac{1}{2}\beta^2(p)\partial_{pp}a_2$ 
with \(a_2(T,p)=0\), which is, thanks again to \eqref{hat u}, the Cauchy problem \eqref{Cauchy 1}; namely, $(g_2,\hat u)$ indeed satisfies \eqref{3-6}. Plugging $(\cV_2, g_2)$ in \eqref{V_2 and g_2} into \eqref{3-5} yields
\begin{align}\label{3-7}
\hspace{-0.05in}0 &= \partial_t A_2 + \sup_{\mathfrak u\in\mathbb R}\Big\{(\theta(p)-r)\mathfrak u +\frac{1}{2}\beta^2(p)\partial_{pp}A_2 -\frac{\gamma_2}{2}\sigma^2\mathfrak u^2 -\frac{\gamma_2}{2}\beta^2(p)(\partial_pa_2)^2 -\gamma_2\mathfrak u\sigma\beta(p)\partial_pa_2 \Big\}\notag\\
&= \partial_t A_2 + \sup_{\mathfrak u\in\mathbb R}\Big\{(\theta(p)-r)\mathfrak u +\frac{1}{2}\beta^2(p)\partial_{pp}A_2 -\frac{\gamma_2}{2} \big(\sigma\mathfrak u +\beta(p)\partial_pa_2\big)^2  \Big\}
\end{align}
with \(A_2(T,p)=0\). By a direct calculation, the maximizer of the supremum in \eqref{3-7} is exactly $\hat u$ in \eqref{hat u}. Taking $\mathfrak u = \hat u(t,p)$ in \eqref{3-7} gives the Cauchy problem \eqref{Cauchy 2}, i.e., $(\cV_2, g_2)$ satisfies \eqref{3-5}.

(iii) For any $(t,p)\in[0,T)\times (0,1)$ and $v\in\R$, consider
\begin{align}
\mathsf C_1(t,p;v) &:=  \partial_t a_2(t,p) + (\theta(p)-r) v + \frac{1}{2}\beta(p)^2\partial_{pp}a_2(t,p),\label{sfC_1}\\
\mathsf C_2(t,p;v) &:= \partial_t A_2(t,p) + \frac{1}{2}\beta^2(p)\partial_{pp}A_2(t,p) +  (\theta(p)-r) v-  \frac{\gamma_2}{2} \big( \sigma v + \beta(p)\partial_p a_2(t,p)\big)^2.\notag
\end{align}
We have $\mathsf C_1(t,p;\hat u(t,p))=0$ and $\mathsf C_2(t,p;\hat u(t,p))=0$, as they are the first lines in  \eqref{Cauchy 1} and  \eqref{Cauchy 2}. Recall from part (i) that the first line in \eqref{Cauchy 2} is obtained by taking $\mathfrak u=\hat u(t,p)$ in \eqref{3-7}. Thus, 
\begin{equation}\label{C_2<=}
\mathsf C_2(t,p;v)\le \mathsf C_2(t,p;\hat u(t,p))=0\quad \forall v\in\R. 
\end{equation}
Now, as $\hat u^{(\alpha)} = \hat u$ (by part (i)), $\hat Z_2$ no longer depends on $\alpha$ and fulfills the dynamics 
\begin{equation}\label{hat Z_2}
d\hat Z_2(s) = \big(\theta(P(s))-r\big)\hat u(s,P(s))\,ds + \sigma \hat u(s,P(s))\,d\widehat W(s),\quad \widehat Z_2(t)=z_2. 
\end{equation}
Under the dynamics of $\hat Z_2$ and $P$ in \eqref{hat Z_2} and \eqref{dynamics investor 2^alpha}, applying It\^o's formula to $g_2\big(s,\hat Z_2(s),P(s)\big) = \hat Z_2(s) + a_2(s,P(s))$ yields
\begin{align}\label{Ito's to g_2}
d \big(g_2\big(s,\hat Z_2(s),P(s)\big)\big) &= \mathsf C_1(s, P(s);\hat u(s,P(s)))\, ds\notag\\
 &\hspace{0.4in}+ \big[\sigma \hat u(s,P(s)) +\beta(P(s))\partial_p a_2(s,P(s))\big]d\widehat W(s).
\end{align}
As $\mathsf C_1(s, P(s);\hat u(s,P(s))) = 0$ and $\hat u$, $\beta$, and $\partial_p a_2$ are bounded (by part (i), \eqref{theta}, and Lemma~\ref{lem:a_2 bdd}), $g_2\big(s,\hat Z_2(s),P(s)\big)$ is a martingale. Hence, 
$g_2(t,z_2,p) = \mathbb E_{t,z_2,p}\big[g_2\big(T, \hat Z_2(T),P(T)\big)\big] = \mathbb E_{t,z_2,p}\big[\hat Z_2(T)\big]$,
where the last equality follows from \eqref{V_2 and g_2} and the terminal condition in \eqref{Cauchy 1}. This proves the first assertion in (iii). Now, by applying It\^o's formula to $\cV_2(s,\hat Z_2(s),P(s)) - \frac{\gamma_2}{2} g_2^2(s,\hat Z_2(s),P(s)) = \hat Z_2(s) + A_2(s,P(s)) - \frac{\gamma_2}{2} (\hat Z_2(s) + a_2(s,P(s)))^2$, we see that the coefficient of the $ds$-term is 
\[
\mathsf C_2(s,P(s);\hat u(s,P(s))) - \gamma_2\big(\hat Z_2(s)+a_2(s,P(s))\big) \mathsf C_1(s,P(s);\hat u(s,P(s)))=0,
\] 
while the coefficient of the $d\widehat W(s)$-term is 
\begin{align}\label{dW-term}
\sigma &\hat u(s,P(s)) \Big[1-\gamma_2\Big(\hat Z_2(s)+a_2(s,P(s))\Big)\Big]\notag\\
&+ \Big[\partial_p A_2(s,P(s)) - \gamma_2\Big(\hat Z_2(s)+a_2(s,P(s))\Big)\partial_p a_2(s,P(s))\Big] \beta(P(s)).
\end{align}
To show that $\cV_2(s,\hat Z_2(s),P(s)) - \frac{\gamma_2}{2} g_2^2(s,\hat Z_2(s),P(s))$ is a martingale, given that $\hat u$, $\beta$, $a_2$,  $\partial_p a_2$, and $\partial_p A_2$ are bounded (by part (i), \eqref{theta}, and Lemma~\ref{lem:a_2 bdd}), it suffices that $\E_{t,z_2,p}\big[\int_t^T |\hat Z_2(s)|^2\, ds\big]<\infty$. As the coefficients in \eqref{hat Z_2} are bounded, \cite[Problem 5.3.15]{KS-book-91} asserts 
$\E_{t,z_2,p}\big[\max_{t\le s\le T}|\hat Z_2(s)|^2\big] \le C (1+z_2^2)$ for some $C>0$, which implies $\E_{t,z_2,p}\big[\int_t^T |\hat Z_2(s)|^2\, ds\big]< \infty$. 
By the martingale property,
\begin{align*}
\cV_2(t,z_2,p) - \frac{\gamma_2}{2}g_2^2(t,z_2,p) &=  \mathbb E_{t,z_2,p}\left[ \cV_2(T,\hat Z_2(T),P(T)) - \frac{\gamma_2}{2} g_2^2(T,\hat Z_2(T),P(T)) \right]\\
&= \mathbb E_{t,z_2,p}\left[\hat Z_2(T)-\frac{\gamma_2}{2}\big(\hat Z_2(T)\big)^2\right],
\end{align*}
where the second line follows from \eqref{V_2 and g_2} and the terminal conditions in \eqref{Cauchy 1} and \eqref{Cauchy 2}. This, along with the first assertion in (iii), implies
\begin{align*}
\cV_2(t,z_2,p) &=  \mathbb E_{t,z_2,p}\left[\hat Z_2(T)-\frac{\gamma_2}{2}\big(\hat Z_2(T)\big)^2\right] + \frac{\gamma_2}{2}\left(\mathbb E_{t,z_2,p}\left[ \hat Z_2(T)\right] \right)^2\\
&=  {\mathbb E}_{t,z_2,p} \left[\hat Z_2(T)\right]   -  \frac{\gamma_2}{2}{\operatorname{Var}}_{t,z_2,p} \left[\hat Z_2(T)\right] = \cJ_2^{(\alpha)}(t,z_2,p;\hat u),\ \ \forall \alpha\in\operatorname{Step}([0,T];\R),
\end{align*}
where the last equality follows from \eqref{cJ_2^alpha} and $\hat u^{(\alpha)}=\hat u$ for all $\alpha\in\operatorname{Step}([0,T];\R)$ (by (i)). 

(iv) Let $\hat Z^{h,v}_2$ satisfy \eqref{Z_2} with $u = \hat u^{h,v}$. Similarly to the calculation in part (ii), when we apply It\^{o}'s formula to $\cV_2(s,\hat Z_2^{h,v}(s),P(s)) - \frac{\gamma_2}{2} g_2^2(s,\hat Z_2^{h,v}(s),P(s))$ on $[t,t+h]$, the $ds$-term is 
\[
\mathsf C_2(s,P(s);v) - \gamma_2\big(\hat Z_2^{h,v}(s)+a_2(s,P(s))\big) \mathsf C_1(s,P(s);v)\le - \gamma_2 g_2\big(s,\hat Z_2^{h,v}(s),P(s)\big) \mathsf C_1(s,P(s);v),
\] 
where the inequality follows from \eqref{C_2<=} and \eqref{V_2 and g_2}, and the $d\widehat W$-term is \eqref{dW-term} with $\hat u(s,P(s))$ and $\hat Z_2$ therein replaced by $v$ and $\hat Z_2^{h,v}$. Hence, by the same argument below \eqref{dW-term}, we get
\begin{align}
&\hspace{0.2in}\E_{t,z_2,p}\left[\cV_2\big(t+h,\hat Z_2^{h,v}(t+h),P(t+h)\big)\right] - \frac{\gamma_2}{2} \E_{t,z_2,p}\left[g_2^2(t+h,\hat Z_2^{h,v}(t+h),P(t+h))\right]\notag\\
&\le \cV_2(t,z_2,p) - \frac{\gamma_2}{2} g_2^2(t,z_2,p)-\gamma_2\E_{t,z_2,p}\left[\int_t^{t+h} g_2\big(s,\hat Z_2^{h,v}(s),P(s)\big) \mathsf C_1(s,P(s);v)\, ds\right].\label{E[V_2]<=}
\end{align}
By the second assertion in (iii), $\cV_2\big(t+h,\hat Z_2^{h,v}(t+h),P(t+h)\big)$ equals
\[
\E_{t+h,\hat Z_2^{h,v}(t+h),P(t+h)}\left[\hat Z_2(T)\right] - \frac{\gamma_2}{2} {\operatorname{Var}}_{t+h,\hat Z_2^{h,v}(t+h),P(t+h)} \left[\hat Z_2(T)\right].
\]
Hence, by the tower property of conditional expectations and the law of total variance, the first term on left-hand side of \eqref{E[V_2]<=} becomes
\begin{align*}
\E_{t,z_2,p}\left[\hat Z_2^{h,v}(T)\right] - \frac{\gamma_2}{2} {\operatorname{Var}}_{t,z_2,p} \left[\hat Z_2^{h,v}(T)\right] +  \frac{\gamma_2}{2} {\operatorname{Var}}_{t,z_2,p} \left[ \E_{t+h,\hat Z_2^{h,v}(t+h),P(t+h)}\big[\hat Z_2(T)\big]\right].
\end{align*}
This, along with $\E_{t+h,\hat Z_2^{h,v}(t+h),P(t+h)}\big[\hat Z_2(T)\big] = g_2\big(t+h,\hat Z_2^{h,v}(t+h),P(t+h)\big)$ (by the first assertion in (iii)), shows that the left-hand side of \eqref{E[V_2]<=} equals
\[
\E_{t,z_2,p}\left[\hat Z_2^{h,v}(T)\right] - \frac{\gamma_2}{2} {\operatorname{Var}}_{t,z_2,p} \left[\hat Z_2^{h,v}(T)\right] - \frac{\gamma_2}{2} \left(\E_{t,z_2,p}\left[g_2\big(t+h,\hat Z_2^{h,v}(t+h),P(t+h)\big)\right]\right)^2.
\]
Hence, \eqref{E[V_2]<=} yields
\begin{align}
&\E_{t,z_2,p}\left[\hat Z_2^{h,v}(T)\right] - \frac{\gamma_2}{2} {\operatorname{Var}}_{t,z_2,p} \left[\hat Z_2^{h,v}(T)\right] -\cV_2(t,z_2,p)\notag\\
&\le \frac{\gamma_2}{2} \bigg\{ \left(\E_{t,z_2,p}\left[g_2\big(t+h,\hat Z_2^{h,v}(t+h),P(t+h)\big)\right]\right)^2 - g_2^2(t,z_2,p)\notag\\
&\hspace{0.75in}-2\E_{t,z_2,p}\left[\int_t^{t+h} g_2\big(s,\hat Z_2^{h,v}(s),P(s)\big) \mathsf C_1(s,P(s);v)\, ds\right]\bigg\}. \label{second last}
\end{align}
By \eqref{cJ_2^alpha} and the second assertion in (iii), the left-hand side of \eqref{second last} equals $\cJ_2^{(\alpha)}(t,z_2,p;\hat u^{h,v}) - \cJ_2^{(\alpha)}(t,z_2,p;\hat u)$ for all $\alpha\in\operatorname{Step}([0,T];\R)$. Also, when we apply It\^{o}'s formula to $g_2\big(s,\hat Z_2^{h,v}(s),P(s)\big)$ on $[t,t+h]$, we get \eqref{Ito's to g_2} with $\hat u(t,p)$ replaced by $v$. By Lemma~\ref{lem:a_2 bdd}, both the drift and volatility coefficients, i.e., $\mathsf C_1(s,p;v)$ in \eqref{sfC_1} and $\mathsf D(s,p;v) := \sigma v+\beta(p)\partial_p a_2(s,p)$, are bounded, so that
\begin{equation}\label{sfM}
\mathsf M:= \|\mathsf C_1(\cdot,\cdot;v)\|_\infty+\|\mathsf D(\cdot,\cdot;v)\|_\infty \le C_0(1+v^2) <\infty,
\end{equation}
where $C_0>0$ depends on only $\sigma$, $\|\theta\|_\infty$, $\|\beta\|_\infty$, $\|\partial_t a_2\|_\infty$, $\|\partial_p a_2\|_\infty$, and $\|\partial_{pp} a_2\|_\infty$. There are two consequences of this. The first one is 
\[
\E_{t,z_2,p}\left[g_2\big(t+h,\hat Z_2^{h,v}(t+h),P(t+h)\big)\right] = g_2(t,z_2,p) +\E_{t,z_2,p}\left[ \int_t^{t+h} \mathsf C_1(s,P(s);v)\, ds\right],  
\]
which implies
\begin{align}
&\left(\E_{t,z_2,p}\left[g_2\big(t+h,\hat Z_2^{h,v}(t+h),P(t+h)\big)\right]\right)^2
= g^2_2(t,z_2,p)\notag \\
&\hspace{0.2in}+2 g_2(t,z_2,p)\E_{t,z_2,p}\left[ \int_t^{t+h} \mathsf C_1(s,P(s);v)\, ds\right]
+ \left(\E_{t,z_2,p}\left[ \int_t^{t+h} \mathsf  C_1(s,P(s);v)\, ds\right]\right)^2,\label{first o(h)}
\end{align}
where the last term is of the order $\mathsf M^2 h^2$. We can thus rewrite \eqref{second last} as 
\begin{align}
&\cJ_2^{(\alpha)}(t,z_2,p;\hat u^{h,v}) - \cJ_2^{(\alpha)}(t,z_2,p;\hat u)\notag\\
&\le \gamma_2 \E_{t,z_2,p}\left[\int_t^{t+h} \Big(g_2(t,z_2,p) - g_2\big(s,\hat Z_2^{h,v}(s),P(s)\big)\Big) \mathsf C_1(s,P(s);v)\, ds\right] + \frac{\gamma_2}{2}\mathsf M^2h^2. \label{last}
\end{align}
The second consequence is
\begin{align}
\E \left[\left|g_2(t,z_2,p) - g_2\big(s,\hat Z_2^{h,v}(s),P(s)\big)\right|^2\right]&\le 2 \E\bigg[ \left(\int_t^s \mathsf C_1(\ell,P(\ell);v) d\ell\right)^2 +   \left(\int_t^s \mathsf D(\ell,P(\ell)) d\widehat W(\ell)\right)^2 \bigg]\notag\\
&\le 2\big(\mathsf M^2 (s-t)^2 + \mathsf M^2 (s-t)\big) \le 2\mathsf M^2  (h+1)h, \label{g-g=o(h)}
\end{align}
where the second inequality follows from It\^{o}'s isometry. Let $I$ denote the expectation on the right-hand side of \eqref{last}. By applying the Cauchy--Schwarz inequality to $I$ twice (first to the integral $\int_t^{t+h}$, next to $\E_{t,z_2,p}$) and using \eqref{g-g=o(h)}, we get 
\begin{align}
|I| &\le \left(\E_{t,z_2,p}\left[\int_t^{t+h} \Big|g_2(t,z_2,p) - g_2\big(s,\hat Z_2^{h,v}(s),P(s)\big)\Big|^2\, ds\right] \right)^{\frac12}  \left(\E_{t,z_2,p}\left[\int_t^{t+h} \mathsf C_1^2\big(s,P(s);v\big)\, ds\right] \right)^{\frac12}\notag\\
&\le (\sqrt{2(h+1)} \mathsf M h )(\mathsf M\sqrt{h})\le \sqrt{2(T+1)} \mathsf M^2h^{3/2},\label{last o(h)}
\end{align}
This, along with \eqref{last}, yields $\cJ_2^{(\alpha)}(t,z_2,p;\hat u^{h,v}) - \cJ_2^{(\alpha)}(t,z_2,p;\hat u)\le \gamma_2 \mathsf M^2 \big(\sqrt{2(T+1)} h^{3/2} + {h^2}/{2}\big)$. In view of \eqref{sfM}, this readily implies \eqref{J^alpha-J^alpha<=}. 


\subsection{Proof of Lemma~\ref{lem:hat Pi}}\label{subsec:proof of lem:hat Pi}
(i) Plugging into \eqref{3-1}-\eqref{3-2} the ansatz
\begin{align}\label{eqa-11-1}
\cV_1(t,z_1,p)=z_1+A_1(t,p),\qquad g_1(t,z_1,p)=z_1+a_1(t,p),
\end{align}
for some functions \(A_1\) and \(a_1\) to be determined, we obtain 
\begin{align}\label{3-3}
0 = \partial_t A_1 + \sup_{\nu\in\cP_2(\R)} \bigg\{&(\mu-r)\left(\chi\fm^\nu-\frac{\lambda_1 \hat u}{2-\lambda_2}\right) + \left[\frac{\mu-\mu_2}{\sigma}-\frac{\mu_1-\mu_2}{\sigma}p\right]\beta(p)\partial_p A_1 + \frac12\beta^2(p)\partial_{pp}A_1 \notag\\
&- \frac{\gamma_1}{2}\sigma^2 \bigg[\left(\chi\fm^\nu-\frac{\lambda_1\hat u}{2-\lambda_2}\right)^2 +\chi^2(\fs^\nu)^2\bigg]-\frac{\gamma_1}{2}\beta^2(p)(\partial_p a_1)^2 \notag\\
&-\gamma_1\sigma\left(\chi\fm^\nu-\frac{\lambda_1\hat u}{2-\lambda_2}\right)\beta(p)\partial_p a_1+\lambda_0 H(\nu)\bigg\},
\end{align}
with terminal condition \(A_1(T,p)=0\), and
\begin{align}\label{3-3-1}
0 = \partial_t a_1 + (\mu-r)\left(\chi\fm^{\hat\Pi_t}-\frac{\lambda_1\hat u}{2-\lambda_2}\right)+\left[\frac{\mu-\mu_2}{\sigma}-\frac{\mu_1-\mu_2}{\sigma}p\right]\beta(p)\partial_p a_1+\frac12\beta^2(p)\partial_{pp}a_1,
\end{align}
with terminal condition \(a_1(T,p)=0\). The terms in \eqref{3-3} that involve \(\nu\) are
\begin{align}\label{eq3-8}
(\mu-r)\chi\fm^\nu-\frac{\gamma_1}{2}\sigma^2 \bigg[ \left(\chi\fm^\nu-\frac{\lambda_1\hat u}{2-\lambda_2}\right)^2 + \chi^2(\fs^\nu)^2\bigg]-\gamma_1\sigma\chi\fm^\nu\,\beta(p)\partial_p a_1+\lambda_0H(\nu),
\end{align}
which, except the term $\lambda_0 H(\nu)$, depend on \(\nu\) through its mean \(\fm^\nu\) and variance \((\fs^\mu)^2\). As Gaussian distributions maximize $H(\nu)$ among $\nu\in\cP(\R)$ with fixed mean and variance (see e.g., \cite{cover2006elements}), the maximizer of \eqref{eq3-8} must be Gaussian. We thus take $\nu\sim\mathcal N(\fm^\nu,(\fs^\nu)^2)$ in \eqref{eq3-8}. Noting that 
\begin{equation}\label{H(nu) Gaussian}
H(\nu) = \frac12\log\big(2\pi(\fs^\nu)^2\big) +\frac12\quad\hbox{for}\ \ \nu\sim\mathcal N(\fm^\nu,(\fs^\nu)^2),
\end{equation}
we find that \eqref{eq3-8} becomes the sum of two concave functions, one in $\fm^\nu$ and the other in $(\fs^\nu)^2$, each of which can be maximized explicitly. The maximizer of \eqref{eq3-8} is then
\begin{equation}\label{hat Pi'}
\hat\Pi_t \sim \mathcal N\bigg(\frac{\mu-r}{\sigma^2\gamma_1\chi} + \frac{\lambda_1\hat u(t,p)}{\chi(2-\lambda_2)}- \frac{\beta(p)}
{\chi\sigma}\partial_p a_1(t,p),\, \frac{\lambda_0}{\gamma_1\sigma^2\chi^2}\bigg). 
\end{equation}
Plugging this back into \eqref{3-3-1}, we obtain
\begin{equation}\label{a_1 eqn}
\partial_t a_1 + \dfrac{(\mu-r)^2}{\sigma^2\gamma_1} +  \big[r-\mu_2-(\mu_1-\mu_2)p\big]\dfrac{\beta(p)\partial_p a_1}{\sigma} + \dfrac12\beta^2(p)\partial_{pp}a_1 =0
\end{equation}
with terminal condition $a_1(T,p)=0$. This trivially admits the solution 
\begin{equation}\label{a_1}
a_1(t,p) = \frac{(\mu-r)^2}{\sigma^2\gamma_1}(T-t),
\end{equation}
which does not depend on $p$ (hence \(\partial_p a_1(t,p)=0\)). Thus, $\hat \Pi_t$ in \eqref{hat Pi'} coincides with \eqref{hat Pi}. By taking $\nu$ in \eqref{3-3} to be the maximizer $\hat\Pi_t$ in \eqref{hat Pi} and using \(\partial_p a_1=0\) and \eqref{H(nu) Gaussian}, we obtain
\begin{equation}\label{A_1 eqn}
0=\partial_t A_1+\frac{(\mu-r)^2}{2\sigma^2\gamma_1}+ \left[ \frac{\mu-\mu_2}{\sigma}-\frac{\mu_1-\mu_2}{\sigma}p
\right]\beta(p)\partial_p A_1 + \frac12\beta^2(p)\partial_{pp}A_1 + \frac{\lambda_0}{2}\log\left(\frac{2\pi \lambda_0}{\gamma_1\sigma^2\chi^2}\right)
\end{equation}
with terminal condition \(A_1(T,p)=0\). This trivially admits the solution
\begin{equation}\label{A_1}
A_1(t,p) = \left[\frac{(\mu-r)^2}{2\sigma^2\gamma_1}+\frac{\lambda_0}{2}\log\left(\frac{2\pi \lambda_0}{\gamma_1\sigma^2\chi^2}\right)\right](T-t),
\end{equation}
which does not depend on $p$. Plugging $A_1$ and $a_1$ we found into \eqref{eqa-11-1}, we obtain \eqref{cV_1}-\eqref{g_1}. 

(ii) 
By the definition of $\hat \Pi$ in \eqref{hat Pi},  $\hat \cZ_1$ fulfills
\begin{align}\label{hat cZ_1}
d \hat \cZ_1(s) &= \frac{(\mu-r)^2}{\sigma^2\gamma_1}\, ds + \frac{\mu-r}{\sigma\gamma_1}\, dW(s) + \bigg(\frac{\lambda_0}{\gamma_1}\bigg)^{1/2}\, d\overline W(s),\quad \hat \cZ_1(t)=z_1,
\end{align}
Thanks to \eqref{g_1} and \eqref{hat cZ_1}, It\^o's formula yields $d g_1(s,\hat \cZ_1(s)) = \frac{\mu-r}{\sigma\gamma_1}\, dW(s) + \big(\frac{\lambda_0}{\gamma_1}\big)^{1/2}\, d\overline W(s)$, which implies that $g_1(s,\hat \cZ_1(s))$ is a martingale. Thus,  $g_1(t,z_1)= \E_{t,z_1} [g_1(T,\hat \cZ_1(T))] = \mathbb E_{t,z_1}[\hat \cZ_1(T)]$, where the last equality follows from \eqref{g_1}. Similarly, by \eqref{cV_1}, \eqref{g_1}, and \eqref{hat cZ_1}, when we apply It\^o's formula to $\cV_1(s,\hat \cZ_1(s)) - \frac{\gamma_1}{2}g_1^2(s,\hat \cZ_1(s))$ on $[t,T]$, the coefficient of the $ds$-term is the constant 
\begin{equation}\label{H is constant}
-\lambda_0 \Theta,\quad\hbox{with}\ \Theta := \frac{1}{2}\log\bigg(\frac{2\pi\lambda_0}{\gamma_1\sigma^2\chi^2}\bigg) + \frac12 = H(\hat\Pi_s)\quad \forall s\in[t,T], 
\end{equation}
where the last equality follows from \eqref{hat Pi} and \eqref{H(nu) Gaussian}, while the local martingale term is
\[
\left(1 -\gamma_1\big(\hat Z_1(s)+a_1(s)\big)\right)\, \bigg(\frac{\mu-r}{\sigma\gamma_1}dW(s) + \bigg(\frac{\lambda_0}{\gamma_1}\bigg)^{1/2}\, d\overline W(s)\bigg).
\]
By the same argument below \eqref{dW-term}, this local martingale is a true martingale. Hence, 
\[
\cV_1(t,z_1) - \frac{\gamma_1}{2}g_1^2(t,z_1) -\lambda_0\E_{t,z_1}\bigg[\int_t^T H(\hat\Pi_s)\, ds\bigg] = \E_{t,z_1}\bigg[\cV_1(T,\hat \cZ_1(T)) - \frac{\gamma_1}{2}g_1^2(T,\hat \cZ_1(T))\bigg].
\]
As $g_1(t,z_1)= \E_{t,z_1}[\hat \cZ_1(T)]$ and the right-hand side above equals $\E_{t,z_1}\big[\hat \cZ_1(T)-\frac{\gamma_1}{2}\hat \cZ_1^2(T)\big]$ (by \eqref{cV_1} and \eqref{g_1}), rearranging the above equality yields the desired identity for $\cV_1(t,z_1)$ in (ii). 

(iii) Recall $a_1$ in \eqref{a_1} and $A_1$ in \eqref{A_1}. For any $(t,p)\in[0,T)\times(0,1)$ and $\nu\in\cP_2(R)$, consider
\begin{align*}
\mathfrak C_1(t,p;\nu) &:=  \partial_t a_1 +  (\mu-r)\left(\chi\fm^{\nu}-\frac{\lambda_1\hat u(t,p)}{2-\lambda_2}\right),\\
\mathfrak C_2(t,p;\nu) &:= \partial_t A_1 +  (\mu-r)\left(\chi\fm^{\nu}-\frac{\lambda_1\hat u(t,p)}{2-\lambda_2}\right)\\ &\hspace{0.7in}  - \frac{\gamma_1}{2}\sigma^2 \bigg[\left(\chi\fm^\nu-\frac{\lambda_1\hat u(t,p)}{2-\lambda_2}\right)^2 +\chi^2(\fs^\nu)^2\bigg]+\lambda_0 H(\nu).
\end{align*}
We have $\mathfrak C_1(t,p;\hat \Pi_t)=0$ and $\mathfrak C_2(t,p;\hat \Pi_t)=0$, as they are simply \eqref{a_1 eqn} and \eqref{A_1 eqn} (recalling $\hat u$ in \eqref{hat u}, $\partial_p a_1 =0$, $\partial_p A_1 =0$, and \eqref{H(nu) Gaussian}). Also, as \eqref{A_1 eqn} was obtained by taking $\nu=\hat \Pi_t$ in \eqref{3-3},
\begin{equation}\label{C_2<='}
\mathfrak C_2(t,p;\nu)\le \mathfrak C_2(t,p;\hat \Pi_t)=0\quad \forall \nu\in\cP_2(\R). 
\end{equation}
Let $\hat \cZ^{h,\nu}_1$ satisfy \eqref{cZ_1} with $\Pi = \hat \Pi^{h,\nu}$. Recall $\cV_1$ and $g_1$ in \eqref{cV_1}-\eqref{g_1}. When we apply It\^{o}'s formula to $\cV_1(s,\hat \cZ_1^{h,v}(s)) - \frac{\gamma_1}{2} g_1^2(s,\hat \cZ_1^{h,v}(s))$ on $[t,t+h]$, the $ds$-term is 
\begin{align*}
\mathfrak C_2(s,P(s);\nu) - \lambda_0 H(\nu) & - \gamma_1\big(\hat \cZ_1^{h,\nu}(s)+a_1(s)\big) \mathfrak C_1(s,P(s);\nu)\\
&\le - \lambda_0 H(\nu) - \gamma_1 g_1\big(s,\hat \cZ_1^{h,\nu}(s)\big) \mathfrak C_1(s,P(s);\nu),
\end{align*}
where the inequality follows from \eqref{C_2<='} and \eqref{g_1}, and the local martingale term is 
\[
\sigma \left[(\mu-r) - \gamma_1(\hat \cZ_1^{h,\nu}+a_1(s))\right]\left( \left(\chi\fm^{\nu}-\frac{\lambda_1\hat u(t,P(s))}{2-\lambda_2}\right)\, dW(s) + \chi\fs^\nu\, d\overline W(s) \right).
\]
As $\nu\in\cP_2(\R)$ and $\hat u$ is bounded (Lemma~\ref{lem:hat u} (i)), we may argue as below \eqref{dW-term} to get 
\begin{align*}
&\hspace{0.2in}\E_{t,z_1,p}\left[\cV_1\big(t+h,\hat \cZ_1^{h,\nu}(t+h)\big)\right] - \frac{\gamma_1}{2} \E_{t,z_1,p}\left[g_1^2(t+h,\hat \cZ_1^{h,\nu}(t+h))\right]\notag\\
&\le \cV_1(t,z_1) - \frac{\gamma_1}{2} g_1^2(t,z_1) - \lambda_0 \int_t^{t+h} H(\nu)\, ds -\gamma_1\E_{t,z_1,p}\left[\int_t^{t+h} g_1\big(s,\hat \cZ_1^{h,\nu}(s)\big) \mathfrak C_1(s,P(s);\nu)\, ds\right].
\end{align*}
Now, by applying the same arguments between \eqref{E[V_2]<=} and \eqref{second last}, with the identities for $(\cV_2,g_2)$ in  Lemma~\ref{lem:hat u} (iii) replaced by those for $(\cV_1,g_1)$ in Lemma~\ref{lem:hat Pi} (ii), we obtain 
\begin{align}
&\E_{t,z_1,p}\left[\hat \cZ_1^{h,\nu}(T) + \lambda_0\int_{t+h}^T H(\hat \Pi_s)\, ds\right] - \frac{\gamma_1}{2} {\operatorname{Var}}_{t,z_1,p} \left[\hat \cZ_1^{h,\nu}(T)\right]\notag\\
& \hspace{1.5in}-\cV_1(t,z_1) + \E_{t,z_1,p}\left[\lambda_0 \int_t^{t+h} H(\nu)\, ds\right] \notag\\
&\le \frac{\gamma_1}{2} \bigg\{ \left(\E_{t,z_1,p}\left[g_1\big(t+h,\hat \cZ_1^{h,\nu}(t+h)\big)\right]\right)^2 - g_1^2(t,z_1)\notag\\
&\hspace{1.5in}-2\E_{t,z_1,p}\left[\int_t^{t+h} g_1\big(s,\hat \cZ_1^{h,\nu}(s)\big) \mathfrak C_1(s,P(s);\nu)\, ds\right]\bigg\}, \label{second last'}
\end{align}
which is the counterpart of \eqref{second last} in the present setting. By \eqref{cJ_1} and Lemma~\ref{lem:hat Pi} (ii), the left-hand side above equals $\cJ_1(t,z_1,p;\hat \Pi^{h,\nu}) - \cJ_1(t,z_1,p;\hat \Pi)$. Also, we can repeat the same arguments below \eqref{second last}, with $g_2$, $\mathsf C_1$, and $\mathsf M$ therein replaced by $g_1$, $\mathfrak C_1$, and 
\[
\mathfrak M:= \|\mathfrak C_1(\cdot,\cdot;\nu)\|_\infty+\|\mathfrak D(\cdot,\cdot;\nu)\|_\infty \le C_0 (1+\fm^\nu+\fs^\nu)<\infty, 
\] 
where $\mathfrak D(t,p;\nu) := \sigma\big[\big(\chi \fm^{\nu} - \frac{\lambda_1 \hat u(t,p)}{2-\lambda_2} \big)+\chi\fs^\nu\big]$ and $C_0>0$ is is independent of $(t,\bm x,p)$, $\nu$, and $h$. This shows that the right-hand side of \eqref{second last'} is bounded from above by $\gamma_1 \mathfrak M^2 \big(\sqrt{2(T+1)} h^{3/2} + {h^2}/{2}\big)$. We thus obtain 
$
\cJ_1(t,z_1,p;\hat \Pi^{h,\nu}) - \cJ_1(t,z_1,p;\hat \Pi) \le\gamma_1 \mathfrak M^2 \big(\sqrt{2(T+1)} h^{3/2} + {h^2}/{2}\big),
$
which, thanks to \eqref{J_1=cJ_1}, implies \eqref{J_1(h)-J_1}. 


\subsection{Proof of Lemma~\ref{lem:(J-J)-(J-J)}}\label{subsec:proof of lem:(J-J)-(J-J)}
Fix a time grid $\D=\{t_i\}_{i=0}^n$ of $[0,T]$, $t_i\in\D\setminus\{T\}$, $(\bm x,p)\in\R^2\times (0,1)$, and $\nu\in\cP_2(\R)$. 

{\bf Step 1:} Express $J_1^\D\big(t_i,\bm x,p;\hat \Pi^{h_i,\nu}\big) - J_1^\D\big(t_i,\bm x,p;\hat \Pi\big)$ explicitly. Let $\hat Z_1^\nu$ be the process \eqref{Z_1}, with $u_1(\delta(s))$ induced by $(\hat\Pi^{h_i,\nu}, \D)$, $t=t_i$, and $z_1$ as in \eqref{z_1}. Note that $\hat Z_1^\nu(T) = z_1+ I_1^\nu + I_2$, with 
\begin{align}
I^\nu_1 &:= \int_{t_i}^{t_i+h_i}  \bigg(\chi Y - \frac{\lambda_1\hat u(s,P(s))}{2-\lambda_2} \bigg) \Big((\mu-r)\, ds+ \sigma\, dW(s)\Big),\label{I_1}\\
I_2 &:=  \int_{t_i+h_i}^{T}  \bigg(\chi u_1(\delta(s)) - \frac{\lambda_1\hat u(s,P(s))}{2-\lambda_2} \bigg) \Big((\mu-r)\, ds+ \sigma\, dW(s)\Big),\label{I_2}
\end{align}
where $Y\sim\nu$ is a random variable independent of $W$ and $\{u_1(\delta(s))\}_{s\in[t_i+h_i,T]}$. Define 
\begin{align}
\Delta W &:= W(t_i+h_i)-W(t_i),\qquad G:= \widetilde\E_{t_i,p}[(\Delta W) I_2], \label{Delta W}\\
L&:= {\operatorname{Cov}}_{t_i,p}\bigg(\int_{t_i}^{t_i+h_i}   - \frac{\lambda_1\hat u(s,P(s))}{2-\lambda_2} \Big((\mu-r)\, ds+ \sigma\, dW(s)\Big),\, I_2\bigg). \notag
\end{align}
 As $Y$ is independent of $W$ and $\{u_1(\delta(s))\}_{s\in[t_i+h_i,T]}$ and $W$ has independent increments, 
\begin{align}\label{Cov}
\widetilde{\operatorname{Cov}}_{t_i,p}(I^\nu_1,I_2) &= \widetilde{\operatorname{Cov}}_{t_i,p}\big(Y\chi \sigma \Delta W, I_2\big) +  L =  \chi \sigma \fm^\nu G +  L, 
\end{align}
where we use $Y\sim\nu$ in the last equality. Now, by \eqref{J_1^D=cJ_1^D} and \eqref{cJ_1^D}, 
\begin{align*}
J_1^\D\big(t_i,\bm x,p;\hat \Pi^{h_i,\nu}\big) = \cJ_1^\D\big(t_i,z_1,p;\hat \Pi^{h_i,\nu}\big) = z_1 &+ \widetilde\E_{t_i,p} [I_1^\nu + I_2] + \lambda_0 h_i  H(\nu) + \lambda_0(T-t_{i+1}) \Theta \\
 &- \frac{\gamma_1}{2}\left( \widetilde{\operatorname{Var}}_{t_i,p}(I^\nu_1) +  \widetilde{\operatorname{Var}}_{t_i,p}(I_2) +2 (\chi\sigma\fm^\nu G+ L) \right),
\end{align*}
where the second equality holds by recalling \eqref{H is constant} and \eqref{Cov}. By taking $\nu =\hat\Pi_{t_i}$ and using \eqref{H is constant} in the above relation, we obtain
\begin{align*}
J_1^\D\big(t_i,\bm x,p;\hat \Pi\big)  = z_1 &+ \widetilde\E_{t_i,p} [\hat I_1 + I_2]  +\lambda_0 (T-t_{i}) \Theta \\
&- \frac{\gamma_1}{2}\left( \widetilde{\operatorname{Var}}_{t_i,p}(\hat I_1) +  \widetilde{\operatorname{Var}}_{t_i,p}(I_2) +2 (\chi\sigma\fm^{\hat\Pi_{t_i}} G+ L) \right),
\end{align*}
where we write $\hat I_1$ for $I_1^{\nu}$ with $\nu = \hat\Pi_{t_i}$. By the previous two equations,
\begin{align}
&J_1^\D\big(t_i,\bm x,p;\hat \Pi^{h_i,\nu}\big) - J_1^\D\big(t_i,\bm x,p;\hat \Pi\big)\notag \\
&=  \widetilde\E_{t_i,p} [I_1^\nu - \hat I_1] + \lambda_0 h_i (H(\nu)-\Theta) 
-  \frac{\gamma_1}{2}\left(\widetilde{\operatorname{Var}}_{t_i,p}(I^\nu_1) -  \widetilde{\operatorname{Var}}_{t_i,p}(\hat I_1)\right)- \gamma_1\chi\sigma \big(\fm^\nu-\fm^{\hat\Pi_{t_i}}\big)G.\label{J_1^D diff}
\end{align}

{\bf Step 2:} Express $J_1\big(t_i,\bm x,p;\hat \Pi^{h_i,\nu}\big) - J_1\big(t_i,\bm x,p;\hat \Pi\big)$ explicitly. Like Step 1, let $\hat\cZ_1^{h_i,\nu}$ denote \eqref{cZ_1} with $\Pi = \hat\Pi^{h_i,\nu}$ (see \eqref{Pi^h}), $t=t_i$, and $z_1$ as in \eqref{z_1}. Then, $\hat \cZ_1^{h_i,\nu}(T) = z_1+ \cI_1^\nu + \cI_2$, with 
\begin{align}
\cI^\nu_1 &:= \int_{t_i}^{t_i+h_i}  \bigg(\chi \fm^\nu - \frac{\lambda_1\hat u(s,P(s))}{2-\lambda_2} \bigg) \Big((\mu-r)\, ds+ \sigma\, dW(s)\Big)+\sigma\chi\,\fs^{\nu}\, d\overline W(s),\label{cI_1}\\
\cI_2 &:=  \int_{t_i+h_i}^{T}  \bigg(\chi \fm^{\hat\Pi_s} - \frac{\lambda_1\hat u(s,P(s))}{2-\lambda_2} \bigg) \Big((\mu-r)\, ds+ \sigma\, dW(s)\Big)+\sigma\chi\,\fs^{\hat\Pi_s}\, d\overline W(s)\notag\\
& =  \int_{t_i+h_i}^{T}  \frac{\mu-r}{\sigma^2 \gamma_1}\Big((\mu-r)\, ds+ \sigma\, dW(s)\Big)+ \bigg(\frac{\lambda_0}{\gamma_1}\bigg)^{1/2}\,  d\overline W(s), \notag
\end{align}
where the last equality follows from \eqref{hat Pi}. 
As $W$ and $\overline W$ have independent increments, ${\operatorname{Cov}}(\cI_1^\nu,\cI_2) = 0$. Now, by \eqref{J_1=cJ_1} and \eqref{cJ_1}, 
\begin{align*}
J_1\big(t_i,\bm x,p;\hat \Pi^{h_i,\nu}\big) = \cJ_1\big(t_i,z_1,p;\hat \Pi^{h_i,\nu}\big) = z_1 &+ \E_{t_i,p} [\cI_1^\nu + \cI_2] + \lambda_0 h_i H(\nu) +\lambda_0 (T-t_{i+1}) \Theta \\
 & - \frac{\gamma_1}{2}\big( {\operatorname{Var}}_{t_i,p}(\cI^\nu_1) +  {\operatorname{Var}}_{t_i,p}(\cI_2) \big),
\end{align*}
where the second equality holds by recalling \eqref{H is constant}. Now, let $\hat\cZ_1$ be the process \eqref{cZ_1} with $\Pi = \hat\Pi$, $t=t_i$, and $z_1$ as in \eqref{z_1}, which is simply \eqref{hat cZ_1}. 
Hence, we may write $\hat \cZ_1(T) = z_1+ \hat \cI_1 + \cI_2$, with
\begin{equation}\label{hat cI}
\hat \cI_1 :=  \int_{t_i}^{t_i+h_i}  \frac{\mu-r}{\sigma^2 \gamma_1}\Big((\mu-r)\, ds+ \sigma\, dW(s)\Big)+ \bigg(\frac{\lambda_0}{\gamma_1}\bigg)^{1/2}\,  d\overline W(s).
\end{equation}
Again, ${\operatorname{Cov}}(\hat \cI_1,\cI_2) = 0$, as $W$ and $\overline W$ have independent increments. Then, by \eqref{J_1=cJ_1} and \eqref{cJ_1}, 
\begin{align*}
J_1\big(t_i,\bm x,p;\hat \Pi\big)  
&= z_1 + \E_{t_i,p} [\hat \cI_1 + \cI_2]  +\lambda_0 (T-t_{i}) \Theta 
- \frac{\gamma_1}{2}\big({\operatorname{Var}}_{t_i,p}(\hat \cI_1) +  {\operatorname{Var}}_{t_i,p}(\cI_2) \big).
\end{align*}
This, along with the equation above \eqref{hat cI}, yields 
\begin{align}
J_1\big(t_i,\bm x,p;\hat \Pi^{h_i,\nu}\big) - J_1\big(t_i,\bm x,p;\hat \Pi\big) &=  \E_{t_i,p} [\cI_1^\nu - \hat \cI_1] + \lambda_0 h_i (H(\nu)-\Theta) \notag\\
&\hspace{0.2in}-  \frac{\gamma_1}{2}\left({\operatorname{Var}}_{t_i,p}(\cI^\nu_1) -  {\operatorname{Var}}_{t_i,p}(\hat \cI_1)\right). \label{J_1 diff}
\end{align}
By \eqref{J_1^D diff} and \eqref{J_1 diff},  
\begin{align}\label{(J-J)-(J-J)}
&\left|   \left(J_1^\D\big(t_i,\bm x,p;\hat \Pi^{h_i,\nu}\big) - J_1^\D\big(t_i,\bm x,p;\hat \Pi\big)\right)  - \left(J_1\big(t_i,\bm x,p;\hat \Pi^{h_i,\nu}\big) - J_1\big(t_i,\bm x,p;\hat \Pi\big) \right)   \right|\notag\\
&\hspace{-0.17in}\le \left|\widetilde\E_{t_i,p} [I^\nu_1]- \E_{t_i,p}[\cI_1^\nu]\right| + \left|\widetilde\E_{t_i,p} [\hat I_1]- \E_{t_i,p}[\hat \cI_1]\right| + \frac{\gamma_1}{2} \left| \widetilde{\operatorname{Var}}_{t_i,p}(I^\nu_1)- {\operatorname{Var}}_{t_i,p}(\cI^\nu_1)\right|\notag\\
&\hspace{1.35in}+ \frac{\gamma_1}{2} \left| \widetilde{\operatorname{Var}}_{t_i,p}(\hat I_1)- {\operatorname{Var}}_{t_i,p}(\hat \cI_1)\right| + \gamma_1\chi\sigma \Big|\fm^\nu-\fm^{\hat\Pi_{t_i}}\Big||G|.
\end{align}
In the remaining steps, we will estimate the terms on the right-hand side of \eqref{(J-J)-(J-J)} closely.

{\bf Step 3:} Estimate $\widetilde\E_{t_i,p} [I^\nu_1]- \E_{t_i,p}[\cI_1^\nu]$ and $\widetilde{\operatorname{Var}}_{t_i,p}(I^\nu_1)- {\operatorname{Var}}_{t_i,p}(\cI^\nu_1)$. 
In \eqref{I_1}, as $Y\sim \nu\in\cP_2(\R)$ and $\hat u$ is bounded (Lemma~\ref{lem:hat u} (i)), the $dW(s)$-term is a true martingale. 
By using this martingale property and Fubini's theorem, we obtain
\begin{equation}\label{E^nu-E^nu}
\widetilde\E_{t_i,p} [I^\nu_1] = \E_{t_i,p} \bigg[ \int_{t_i}^{t_i+h_i}  \bigg(\chi \fm^\nu - \frac{\lambda_1\hat u(s,P(s))}{2-\lambda_2} \bigg) (\mu-r)\, ds \bigg] = \E_{t_i,p} [\cI^\nu_1], 
\end{equation}
where the last equality follows again from $\nu\in\cP_2(\R)$ and $\hat u$ being bounded (so that the $dW(s)$- and $d\overline W(s)$-terms in \eqref{cI_1} are true martingales). Now, by writing $I^\nu_1$ \eqref{I_1} as $\mathsf A+\mathsf B +\mathsf C + \mathsf D$, with 
\begin{align}
&\mathsf A := \int_{t_i}^{t_i+h_i}  \chi Y (\mu-r)\, ds, \quad \mathsf B := -\int_{t_i}^{t_i+h_i}  \frac{\lambda_1\hat u(s,P(s))}{2-\lambda_2} (\mu-r)\, ds,\label{A, B}\\
&\mathsf C := \int_{t_i}^{t_i+h_i}  \chi Y \sigma\, dW(s),\quad \mathsf D :=  -\int_{t_i}^{t_i+h_i} \frac{\lambda_1\hat u(s,P(s))}{2-\lambda_2} \sigma\, dW(s),\label{C, D}
\end{align}
we deduce from $Y$ being independent of $W$ that 
\begin{align}\label{V(I^nu)}
\widetilde{\operatorname{Var}}_{t_i,p}(I^\nu_1) &= \widetilde{\operatorname{Var}}(\sfA) + {\operatorname{Var}}_{t_i,p}(\sfB) 
+ \widetilde{\operatorname{Var}}_{t_i,p}(\sfC+\sfD) + 2\widetilde{\operatorname{Cov}}_{t_i,p}(\sfB, \sfC)+ 2{\operatorname{Cov}}_{t_i,p}(\sfB, \sfD),
\end{align}
Similarly, we can write $\cI^\nu_1$ in \eqref{cI_1} as $\sfA'+\sfB+\sfC'+\mathsf D+\mathsf E$, with 
\begin{align*}
&\mathsf A' := \int_{t_i}^{t_i+h_i}  \chi \fm^\nu (\mu-r)\, ds,\quad \mathsf C' := \int_{t_i}^{t_i+h_i}  \chi \fm^\nu\sigma\, dW(s), \quad\mathsf E := \int_{t_i}^{t_i+h_i} \sigma\chi \fs^\nu\, d\overline W(s).
\end{align*}
As $\sfA'$ is a constant and $\overline W(\cdot)$ is independent of $W(\cdot)$ (and thus of $P(\cdot)$), we have 
\begin{align}\label{V(cI^nu)}
{\operatorname{Var}}_{t_i,p}(\cI^\nu_1) &= {\operatorname{Var}}_{t_i,p}(\sfB) + {\operatorname{Var}}_{t_i,p}(\sfC'+\sfD+\mathsf E)+2{\operatorname{Cov}}_{t_i,p}(\sfB, \sfC')+ 2{\operatorname{Cov}}_{t_i,p}(\sfB,\mathsf D). 
\end{align}
Observe from It\^{o}'s isometry that
\begin{align}
\widetilde{\operatorname{Var}}_{t_i,p}&(\sfC+\sfD) =  \int_{t_i}^{t_i+h_i}  \widetilde\E_{t_i,p}\bigg[\bigg(\chi Y - \frac{\lambda_1\hat u(s,P(s))}{2-\lambda_2} \bigg)^2\bigg] \sigma^2\, ds\notag\\
&=   \int_{t_i}^{t_i+h_i} \sigma^2\chi^2 (\fs^\nu)^2 +   \E_{t_i,p}\bigg[\bigg(\chi \fm^\nu - \frac{\lambda_1\hat u(s,P(s))}{2-\lambda_2} \bigg)^2\bigg]\sigma^2\, ds = {\operatorname{Var}}_{t_i,p}(\sfC'+\sfD+\mathsf E), \label{V=V}
\end{align}
where the second equality is due to $Y\sim\nu$ (so that $\widetilde\E[Y^2] = (\fs^\nu)^2 + (\fm^\nu)^2$) and $Y$ is independent of $W(\cdot)$ (and thus of $P(\cdot)$), and the third equality follows again from It\^{o}'s isometry and $\overline W(\cdot)$ being independent of $W(\cdot)$ (and thus of $P(\cdot)$). Moreover, as $\sfC = Y \int_{t_i}^{t_i+h_i}  \chi \sigma\, dW(s)$ and $Y$ is independent of $W(\cdot)$ (and thus of $P(\cdot)$), we have $\widetilde\E[\sfC] =0$ and 
\begin{align}
\widetilde{\operatorname{Cov}}_{t_i,p}(\sfB, \sfC) &= \widetilde\E_{t_i,p}[ \sfB \sfC ] = \widetilde\E[Y] \E_{t_i,p}\bigg[\sfB \int_{t_i}^{t_i+h_i}  \chi \sigma\, dW(s)\bigg]\notag\\
& = \E_{t_i,p}\bigg[\sfB \int_{t_i}^{t_i+h_i}  \chi \fm^\nu\sigma\, dW(s)\bigg] =  \E_{t_i,p}[\sfB \sfC'] = {\operatorname{Cov}}_{t_i,p}(\sfB, \sfC'),\label{Cov=Cov}
\end{align}
where the third equality holds as $Y\sim\nu$ (so that $\widetilde\E[Y] = \fm^\nu$). When we subtract \eqref{V(cI^nu)} from \eqref{V(I^nu)}, we conclude from \eqref{V=V} and \eqref{Cov=Cov} that
\begin{equation}\label{Var-Var}
\widetilde{\operatorname{Var}}_{t_i,p}(I^\nu_1)-{\operatorname{Var}}_{t_i,p}(\cI_1^\nu) = \widetilde{\operatorname{Var}}(\sfA) =\chi^2(\mu-r)^2 h_i^2 \widetilde{\operatorname{Var}}(Y) = \chi^2(\mu-r)^2 (\fs^\nu)^2 h_i^2.
\end{equation}

{\bf Step 4:} Estimate $\widetilde\E_{t_i,p} [\hat I_1]- \E_{t_i,p}[\hat\cI_1]$ and $\widetilde{\operatorname{Var}}_{t_i,p}(\hat I_1)- {\operatorname{Var}}_{t_i,p}(\hat\cI_1)$.
Recall that $\hat I_1$ is the process \eqref{I_1} with $Y\sim \hat\Pi_{t_i}$. Hence, by taking $\nu = \hat\Pi_{t_i}$ in the first equality in \eqref{E^nu-E^nu}, we get 
\begin{align}
\widetilde\E_{t_i,p}[\hat I_1] &= \E_{t_i,p} \bigg[ \int_{t_i}^{t_i+h_i}  \bigg(\chi \fm^{\hat\Pi_{t_i}} - \frac{\lambda_1\hat u(s,P(s))}{2-\lambda_2} \bigg) (\mu-r)\, ds \bigg] \notag\\
&= \E_{t_i,p} \bigg[ \int_{t_i}^{t_i+h_i}  \frac{(\mu-r)^2}{\sigma^2\gamma_1} + \frac{\lambda_1}{2-\lambda_2} \Big( \hat u(t_i,p) - \hat u(s,P(s))\Big) (\mu-r)\, ds \bigg]\notag \\
&= \E_{t_i,p}[\hat \cI_1] + \frac{\lambda_1(\mu-r)}{2-\lambda_2} E_1,\quad \hbox{with}\ \ E_1 := \E_{t_i,p}\bigg[ \int_{t_i}^{t_i+h_i}  \hat u(t_i,p) - \hat u(s,P(s))\, ds \bigg]\label{E=E+}
\end{align}
where the second line follows from \eqref{hat Pi} and the third line is due to \eqref{hat cI}. In view of $\hat u$ in \eqref{hat u} and Lemma~\ref{lem:a_2 bdd}, we may apply It\^{o}'s formula to $\hat u(s,P(s))$ on $[t_i,t_i+h_i]$, where $P$ satisfies the dynamics in \eqref{dynamics investor 1}. As all the involved derivatives (i.e., $\partial_{p}a_2$, $\partial_{pp}a_2$, $\partial_{ppp}a_2$, and $\partial_t\partial_{p}a_2$) and coefficients (i.e., $\theta(\cdot)$, $\beta(\cdot)$, and their derivatives) are bounded, the $ds$-term and $dW$-term are both bounded (say, by some $M>0$ depending on $\|\hat u\|_\infty$, $\|\partial_t\hat u\|_\infty$, $\|\partial_p \hat u\|_\infty$, and $\|\partial_{pp}\hat u\|_\infty$). Hence, 
\begin{equation}\label{|hat u-hat u|}
\E_{t_i,p}\left[\left|\hat u(s,P(s)) - \hat u(t_i,p)\right|\right] \le M h_i,\quad \forall s\in[t_i,t_i+h_i]. 
\end{equation}
Moreover, a calculation similar to \eqref{g-g=o(h)} yields 
\begin{equation}\label{|hat u-hat u|^2}
\E_{t_i,p}\left[\left|\hat u(s,P(s)) - \hat u(t_i,p)\right|^2\right] \le 2M^2 (h_i+1)h_i,\quad \forall s\in[t_i,t_i+h_i].
\end{equation}
Thanks to \eqref{|hat u-hat u|}, $E_1$ in \eqref{E=E+} satisfies 
\begin{equation}\label{|E_1|}
|E_1|\le M h_i^2.
\end{equation} 
Now, to find $\widetilde{\operatorname{Var}}_{t_i,p}(\hat I_1)$, we simply take $\nu = \hat\Pi_{t_i}$ in \eqref{V(I^nu)}. Let us compute all the terms on the right-hand side of \eqref{V(I^nu)} (under $Y\sim \nu = \hat\Pi_{t_i}$). First, we deduce from \eqref{A, B} and \eqref{hat Pi} that 
\begin{equation}\label{VarA}
\widetilde{\operatorname{Var}}(\sfA) =\chi^2(\mu-r)^2 h_i^2 \widetilde{\operatorname{Var}}(Y) =  \frac{(\mu-r)^2\lambda_0}{\gamma_1\sigma^2} h_i^2.
\end{equation}
Next, as ${\operatorname{Var}}_{t_i,p}(\sfB) = \E_{t_i,p}[(\sfB-\E_{t_i,p}[\sfB])^2]$, by \eqref{A, B} and Fubini's theorem, we have
\begin{align}\label{VarB}
{\operatorname{Var}}_{t_i,p}(\sfB) &= \bigg(\frac{\lambda_1(\mu-r)}{2-\lambda_2}\bigg)^2 \E_{t_i,p}\bigg[\left(\int_{t_i}^{t_i+h_i}  \hat u(s,P(s))-\E_{t_i,p}[\hat u(s,P(s))]\, ds \right)^2 \bigg]\notag\\
&\le  4 \bigg(\frac{\lambda_1(\mu-r)}{2-\lambda_2}\bigg)^2  \|\hat u\|^2_\infty h_i^2.  
\end{align}
In addition, by taking $Y\sim\nu = \hat\Pi_{t_1}$ in the first two equalities of \eqref{V=V} and using \eqref{hat Pi}, we get
\begin{align}\label{VarC+D}
\widetilde{\operatorname{Var}}_{t_i,p}(\sfC+\sfD) &=   \int_{t_i}^{t_i+h_i} \sigma^2\chi^2 (\fs^\nu)^2 +   \E_{t_i,p}\bigg[\bigg(\chi \fm^\nu - \frac{\lambda_1\hat u(s,P(s))}{2-\lambda_2} \bigg)^2\bigg]\sigma^2\, ds\, \Bigg|_{\nu = \hat\Pi_{t_i}}\notag\\
& =  \int_{t_i}^{t_i+h_i}\frac{\lambda_0}{\gamma_1} +  \E_{t_i,p}\bigg[\bigg(\frac{\mu-r}{\sigma\gamma_1} +\frac{\sigma\lambda_1}{2-\lambda_2}\Big(\hat u(t_i,p) - \hat u(s,P(s))\Big) \bigg)^2\bigg]\, ds\notag\\
&=  \int_{t_i}^{t_i+h_i}\frac{\lambda_0}{\gamma_1} + \bigg(\frac{\mu-r}{\sigma\gamma_1}\bigg)^2\, ds +E_2 +E_3= {\operatorname{Var}}_{t_i,p}(\hat \cI_1) + E_2+E_3,  
\end{align}
where the last equality follows from \eqref{hat cI} and we define 
\begin{align*}
E_2 := \frac{2(\mu-r)\lambda_1}{\gamma_1(2-\lambda_2)} E_1,\quad 
E_3 := \bigg(\frac{\sigma\lambda_1}{2-\lambda_2}\bigg)^2 \E_{t_i,p}\bigg[\int_{t_i}^{t_i+h_i} \Big(\hat u(t_i,p) - \hat u(s,P(s))\Big)^2\, ds\bigg], 
\end{align*}
with $E_1$ given as in \eqref{E=E+}. Thanks to \eqref{|E_1|} and \eqref{|hat u-hat u|^2}, we have
\begin{equation}\label{|E_2|}
|E_2| \le \frac{2(\mu-r)\lambda_1}{\gamma_1(2-\lambda_2)} M h_i^2,\qquad |E_3|\le 2 \bigg(\frac{\sigma\lambda_1}{2-\lambda_2}\bigg)^2 M^2(h_i+1)h_i^2. 
\end{equation} 
Finally, by taking $Y\sim\nu = \hat\Pi_{t_1}$ in the first line of \eqref{Cov=Cov} and using \eqref{hat Pi}, we obtain
\begin{align}\label{CovB,C}
\left|\widetilde{\operatorname{Cov}}_{t_i,p}(\sfB, \sfC)\right| &= \bigg|\widetilde\E[Y] \E_{t_i,p}\bigg[\sfB \int_{t_i}^{t_i+h_i}  \chi \sigma\, dW(s)\bigg]\bigg|\notag\\
&\le \left|\frac{\mu-r}{\sigma^2 \gamma_1 \chi}+\frac{\lambda_1 \hat u(t_i,p)}{2-\lambda_2-\lambda_1}\right| \E_{t_i,p}\bigg[|\sfB| \bigg| \int_{t_i}^{t_i+h_i}  \chi \sigma\, dW(s)\bigg|\bigg]\notag\\
&\le \left(\frac{|\mu-r|}{\sigma^2 \gamma_1 \chi}+\frac{\lambda_1 \|\hat u\|_\infty}{2-\lambda_2-\lambda_1}\right) \frac{\lambda_1|\mu-r|}{2-\lambda_2}\|\hat u\|_\infty h_i  \bigg(\E_{t_i,p}\bigg[ \bigg| \int_{t_i}^{t_i+h_i}  \chi \sigma\, dW(s)\bigg|^2\bigg]\bigg)^{1/2}\notag\\
&= \left(\frac{|\mu-r|}{\sigma^2 \gamma_1 \chi}+\frac{\lambda_1 \|\hat u\|_\infty}{2-\lambda_2-\lambda_1}\right) \frac{\lambda_1|\mu-r|}{2-\lambda_2}\|\hat u\|_\infty  \chi \sigma h_i^{3/2}, 
\end{align}
where the second inequality follows from the definition of $\sfB$ in \eqref{A, B} and $L^1$-norms being bounded from above by $L^2$-norm, and the last equality is due to It\^{o}'s isometry. In a similar vein, 
\begin{align}\label{CovB,D}
\left|{\operatorname{Cov}}_{t_i,p}(\sfB, \sfD)\right| = \left|\E[\sfB\sfD]\right| \le \E[|\sfB||\sfD|] &\le   \frac{\lambda_1|\mu-r|}{2-\lambda_2}\|\hat u\|_\infty h_i (\E_{t_i,p}[\sfD^2])^{1/2}\notag\\
&\le \bigg(\frac{\lambda_1}{2-\lambda_2}\bigg)^2|\mu-r|\sigma\|\hat u\|^2_\infty h_i^{3/2}, 
\end{align}
where the last inequality stems from It\^{o}'s isometry. 
Recall that $\widetilde{\operatorname{Var}}_{t_i,p}(\hat I_1)$ is the sum of the left-hand sides of \eqref{VarA}, \eqref{VarB}, \eqref{VarC+D}, \eqref{CovB,C}, and \eqref{CovB,D}. Using these estimates and \eqref{|E_2|} yields 
\begin{align}
\left|\widetilde{\operatorname{Var}}_{t_i,p}(\hat I_1) - {\operatorname{Var}}_{t_i,p}(\hat \cI_1)\right| &= \left|\widetilde{\operatorname{Var}}_{t_i,p}(\sfA) + {\operatorname{Var}}_{t_i,p}(\sfB) + E_2 + E_3 + \widetilde{\operatorname{Cov}}_{t_i,p}(\sfB, \sfC)+{\operatorname{Cov}}_{t_i,p}(\sfB,\sfD)\right|\notag\\
&\le  C_1 h_i^{3/2} + C_2 h_i^2, \label{Var-Var'}
\end{align}
where $C_1>0$ depends on only $\sigma$, $|\mu-r|$, $\gamma_1$, $\lambda_1$, $\lambda_2$, $\|\hat u\|_\infty$ and $C_2>0$ depends on only $\sigma$, $|\mu-r|$, $\gamma_1$, $\lambda_0$, $\lambda_1$, $\lambda_2$, $\|\hat u\|_\infty$, $\|\partial_t\hat u\|_\infty$, $\|\partial_p \hat u\|_\infty$, and $\|\partial_{pp}\hat u\|_\infty$. 

{\bf Step 5:} Estimate $G =  \widetilde\E_{t_i,p}[(\Delta W)  I_2]$; recall \eqref{Delta W}. First, observe that
\begin{align}
G &=   \widetilde\E_{t_i,p}\left[\Delta W\    \E_{t_i,p}\left[ I_2\mid\cF_{t_i+h_i}\right] \right]\notag\\
&=  \widetilde\E_{t_i,p}\left[\Delta W\ \int_{t_i+h_i}^{T} \bigg(\chi u_1(\delta(s)) -\frac{\lambda_1\hat u(s,P(s))}{2-\lambda_2}\bigg)  (\mu-r)\, ds  \right]\notag\\
&= (\mu-r) \sum_{t_k\in \D\cap (t_i,T)} \widetilde\E_{t_i,p}\left[\Delta W \int_{t_k}^{t_{k+1}} \bigg(\chi u_1(t_k) -\frac{\lambda_1\hat u(s,P(s))}{2-\lambda_2}\bigg)\, ds  \right]\notag\\
&= (\mu-r) \sum_{t_k\in \D\cap (t_i,T)} \E_{t_i,p}\left[\Delta W \int_{t_k}^{t_{k+1}} \bigg(\chi \fm^{\hat\Pi_{t_k}} -\frac{\lambda_1\hat u(s,P(s))}{2-\lambda_2}\bigg)\, ds  \right]\notag\\
&= \frac{\lambda_1(\mu-r)}{2-\lambda_2} \sum_{t_k\in \D\cap (t_i,T)} \int_{t_k}^{t_{k+1}}\E_{t_i,p}\left[\Delta W \Big(\hat u(t_k,P(t_k))-\hat u(s,P(s))\Big)  \right]\, ds,\label{G-1}
\end{align}
where the second equality holds as the $dW(s)$-term in $I_2$ (see \eqref{I_2}) forms a martingale starting from 0 on $[t_i+h_i,T]$ (disappearing after taking the conditional expectation), the third equality results from the definition of $\delta$ (in the line below \eqref{eq2-14}), the fourth equality holds by Fubini's theorem and $u_1(t_k)\sim\hat\Pi_{t_k}$, and the last line follows from $\chi \fm^{\hat\Pi_{t_k}} = \frac{\mu-r}{\sigma^2\gamma_1} +  \frac{\lambda_1\hat u (t_k,P(t_k))}{2-\lambda_2}$ (by \eqref{hat Pi}), Fubini's theorem, and $\E_{t_i,p}[\Delta W \frac{\mu-r}{\sigma^2\gamma_1}] = \frac{\mu-r}{\sigma^2\gamma_1} \E_{t_i,p}[\Delta W]  = 0$. 
Now, for each $t_k\in\D\cap (t_i,T)$ and $s\in [t_k, t_{k+1}]$, by integration by parts in Malliavin calculus (i.e., \cite[eqn.\ (1.42)]{Nualart-book-06}),
\begin{align}\label{Cov t_k}
\E_{t_i,p}\left[\Delta W  \Big(\hat u(t_k,P(t_k))-\hat u(s,P(s))\Big)\right] = \int_{t_i}^{t_i+h_i} \E_{t_i,p}\bigg[\mathbb D_t  \Big(\hat u(t_k,P(t_k))-\hat u(s,P(s))\Big)\bigg]\, dt,
\end{align}
where $\mathbb D_t$ is the Malliavin derivative operator w.r.t.\ a perturbation in $W(t)$; specifically, in \cite[eqn.\ (1.42)]{Nualart-book-06}, we take the random variable $F$ to be $\hat u(t_k,P(t_k))-\hat u(s,P(s))$ and the process $u$ to be $1_{[t_i,t_i+h_i]}(s)$, so that $\delta(u) = \delta(1_{[t_i,t_i+h_i]}) = \int_{0}^{T} 1_{[t_i,t_i+h_i]}(s)\, dW(s)=\Delta W$ and $\langle \mathbb D F, u\rangle_H = \int_{t_i}^{t_i+h_i} \mathbb D_t \big(\hat u(t_k,P(t_k))-\hat u(s,P(s))\big)\, ds$. By the Malliavin chain rule (see \cite[Proposition 1.2.3]{Nualart-book-06}), 
\begin{align}\label{D_t}
\mathbb D_t  \Big(\hat u(t_k,P(t_k))-\hat u(s,P(s))\Big) &= \partial_p \hat u(t_k,P(t_k)) \mathbb D_t P(t_k)-\partial_p \hat u(s,P(s)) \mathbb D_t P(s). 
\end{align}
Furthermore, by \cite[Section 2.2]{Nualart-book-06} (particularly eqn.\ (2.53) therein), 
\begin{equation}\label{D_t P}
\mathbb D_t P(s) = \beta(P(t))\zeta(t)^{-1}\zeta(s),\quad \forall s\in[t_k,t_{k+1}],
\end{equation}
 where $\zeta$ is the unique strong solution to 
\begin{equation}\label{zeta'}
d{\zeta}(s) = \zeta(s)\Gamma(P^{t,p}(s))\, ds + \zeta(s)\Lambda(P^{t,p}(s))\, d W(s),\quad \zeta(t_i)=1,
\end{equation}
with $\Gamma(p) := \frac{d}{dp}(\left[\frac{\mu-\mu_2}{\sigma}-\frac{\mu_1-\mu_2}{\sigma}p\right]\beta(p))$  and $\Lambda(p) := \frac{d}{dp}\beta(p)$ (i.e., \eqref{zeta} is updated to reflect the $P$-dynamics given in \eqref{dynamics investor 1}). By \eqref{theta}, $\beta$, $\Gamma$, and $\Lambda$ are all bounded on $(0,1)$. This implies that for any $m\in\R$ and $t\in[t_i,T]$, $0\le \E_{t_i,p}[\zeta^m(t)]\le K_m$ for some $K_m>0$ depending on only $\|\Gamma\|_\infty$, $\|\Lambda\|_\infty$, $m$, and $T$. Now, by plugging \eqref{D_t P} into \eqref{D_t} and using the Cauchy--Schwarz inequality, 
\begin{align}
&\E_{t_i,p}\Big[\Big|\mathbb D_t  \Big(\hat u(t_k,P(t_k))-\hat u(s,P(s))\Big)\Big|\Big] \\
&\le \|\beta\|_\infty (\E_{t_i,p}[\zeta^{-2}(t)])^{1/2} \bigg(\E_{t_i,p}\Big[\Big(\partial_p \hat u(t_k,P(t_k)) \zeta(t_k) -  \partial_p \hat u(s,P(s)) \zeta(s)\Big)^2\Big]\bigg)^{1/2}\\
&\le \|\beta\|_\infty (K_{-2})^{1/2} (2\widetilde M^2 (h_k+1)h_k)^{1/2}, 
\end{align}
where the last inequality follows from applying It\^{o}'s formula to $\partial_p \hat u(s,P(s)) \zeta(s)$ on $[t_k, t_{k+1}]$ and using an estimate similar to \eqref{g-g=o(h)}, with $\mathsf M$ therein replaced by $\widetilde M>0$ depending on only $\|\hat u\|_\infty$, $\|\partial_t\partial_p\hat u\|_\infty$, $\|\partial_p \hat u\|_\infty$, $\|\partial_{pp}\hat u\|_\infty$, $\|\partial_{ppp}\hat u\|_\infty$ (which are all finite, in view of \eqref{hat u} and Lemma~\ref{lem:a_2 bdd}), $K_1$, and $K_2$. Plugging the above estimate into \eqref{Cov t_k} yields $\big|\E_{t_i,p}\big[\Delta W  \big(\hat u(t_k,P(t_k))-\hat u(s,P(s))\big)\big]\big|\le C h_k^{1/2} h_i\le C |\D|^{3/2}$, with $C:= \|\beta\|_\infty \sqrt{2 (T+1) K_{-2}} \widetilde M$. This, along with \eqref{G-1}, yields
\begin{equation}\label{G-final}
|G| \le \frac{\lambda_1|\mu-r|}{2-\lambda_2}  C |\D|^{3/2} \sum_{t_k\in \D\cap (t_i,T)} h_k \le \frac{\lambda_1|\mu-r|}{2-\lambda_2} C T |\D|^{3/2}. 
\end{equation}

{\bf Step 6:} Combine all the estimates. By plugging \eqref{E^nu-E^nu}, \eqref{Var-Var}, \eqref{E=E+}, \eqref{Var-Var'}, and \eqref{G-final} into \eqref{(J-J)-(J-J)} and using \eqref{|E_1|} and $h_i\le |\D|$, we get \eqref{(J-J)-(J-J) to prove}. 

\small
\bibliographystyle{siam}
\bibliography{control}

\end{document}